\documentclass[12pt,draftclsnofoot,onecolumn,a4paper]{IEEEtran}

\usepackage{longtable}
\usepackage{graphicx}
\usepackage{subfigure}
\usepackage{epstopdf}
\usepackage{citesort}
\usepackage{bm}
\usepackage{amssymb}
\usepackage{amsthm}
\usepackage{amsfonts}
\usepackage{amsmath}
\usepackage{epsfig}
\usepackage{fancybox}
\usepackage{textcomp}
\usepackage{multirow}
\usepackage{booktabs}
\usepackage{mathrsfs}
\usepackage{pifont}
\usepackage{bbm}
\usepackage[ruled]{algorithm2e}
\usepackage{euscript}
\usepackage{xcolor}

\usepackage{comment}

\newtheorem{theorem}{Theorem}

\newtheorem{corollary}{Corollary}

\newlength{\halfpagewidth}
\divide\halfpagewidth by 2

\makeatletter
\def\ScaleIfNeeded{%
\ifdim\Gin@nat@width>\linewidth \linewidth \else \Gin@nat@width
\fi } \makeatother

\begin{document}
%\pagestyle{fancyplain}
%
%\pagestyle{fancy}
%\lhead[]%
 %   {\footnotesize Physical layer security}
%\cfoot{}

\small{ \title{{Secrecy and Energy Efficiency in Massive MIMO Aided Heterogeneous C-RAN: A New Look at Interference }}}

\author{Lifeng Wang,~\IEEEmembership{Member,~IEEE,}  Kai-Kit Wong,~\IEEEmembership{Fellow,~IEEE,} Maged Elkashlan,~\IEEEmembership{Member,~IEEE,}  Arumugam Nallanathan,~\IEEEmembership{Senior Member,~IEEE,}  and Sangarapillai Lambotharan,~\IEEEmembership{Senior Member,~IEEE}
\thanks{ L. Wang, and K.-K. Wong are with the Department of Electronic and Electrical Engineering, University College London (UCL), London, UK (Email: $\rm\{lifeng.wang, kai$-$\rm kit.wong\}@ucl.ac.uk$).}
\thanks{M. Elkashlan is with  the School of Electronic Engineering and Computer Science, Queen Mary University of London, London, UK. (Email: $\rm\{maged.elkashlan\}@qmul.ac.uk$)}
\thanks{A. Nallanathan is with the Department of Informatics, King's College London, London, UK (Email: $\rm\{arumugam.nallanathan\}@kcl.ac.uk$)}
\thanks{S. Lambotharan is with School of Electronic, Electrical and System Engineering at Loughborough University, Loughborough Leicestershire, UK (Email: $\rm\{s.lambotharan\}@lboro.ac.uk$)}
}

\maketitle
\vspace{-1.1 cm}
\begin{abstract}
In this paper, we investigate the potential benefits of the massive multiple-input multiple-output (MIMO) enabled heterogeneous cloud radio access network (C-RAN) in terms of the secrecy and energy efficiency (EE). In this network, both remote radio heads (RRHs) and massive MIMO macrocell base stations (BSs) are deployed and soft fractional frequency reuse (S-FFR) is adopted to mitigate the inter-tier interference. We first examine the physical layer security by deriving the area ergodic secrecy rate and secrecy outage probability. Our results reveal that the use of massive MIMO  and C-RAN can greatly improve the secrecy performance. For C-RAN, a large number of RRHs achieves high area ergodic secrecy rate and low secrecy outage probability, due to its powerful interference management. We find that for massive MIMO aided macrocells, having more antennas and serving more users improves secrecy performance. Then we derive the EE of the heterogeneous C-RAN, illustrating that increasing the number of RRHs significantly enhances the network EE. Furthermore, it is indicated that {allocating more radio resources} to the RRHs can linearly increase the EE of RRH tier and improve the network EE without affecting the EE of the macrocells.
\end{abstract}
\vspace{-0.7 cm}

\newpage
%========================================================================
\section{Introduction}
As a new mobile network architecture consisting of remote radio heads (RRHs) and baseband units (BBUs), cloud radio access network (C-RAN) can  deal with large-scale control/data processing much more efficiently. The rationale behind this is that baseband processing is centralized and coordinated among sites in the centralized BBU pool, which reduces the capital expenditure (CAPEX) and operating expenditure (OPEX) of the networks~\cite{checko_2015}. Massive multiple-input multiple-output (MIMO)  is another key technology that promises outstanding spectral efficiency (SE) and energy efficiency (EE). In massive MIMO antenna systems, base stations (BSs) are equipped with large antenna arrays to support a large number of users in the same time-frequency domain~\cite{ngo2013energy}. Among other emerging technologies such as {device-to-device communications}, full duplex radios, and millimeter wave, etc.,  C-RAN and massive MIMO are identified as promising 5G technologies~\cite{Jeffrey_5G,HE15,dantong-survey}.

Driven by its high SE and EE, C-RAN has recently received tremendous attention from both industry and academia~\cite{mugen_poor_2015,ZhiguoDing_2013}. For instance, a group of single-antenna RRHs were considered to form a distributed antenna array, and two downlink transmission strategies namely best RRH selection and distributed beamforming were examined in terms of outage probability in \cite{ZhiguoDing_2013}. Most recently in \cite{Zaidi_2015}, user-centric association in a multi-tier C-RAN was proposed, in which the RRH that had the best signal-to-noise ratio (SNR) was scheduled to serve the user. Compared to  \cite{ZhiguoDing_2013}, downlink transmission in the C-RAN with a group of multi-antenna RRHs was investigated in \cite{Khan_2015}. In the work of \cite{Khan_2015}, maximal ratio transmission and transmit antenna selection were adopted at the RRHs, and the outage probability was derived by considering several transmission schemes such as RRH selection and distributed beamforming.

Heterogeneous C-RAN is a new paradigm by integrating cloud computing with heterogeneous networks (HetNets)~\cite{Mugen_Peng_2014_mag,Peng_M_2015}. In heterogeneous C-RAN, severe inter-tier interference is coordinated for the enhancement of SE and EE.  The architecture of heterogenous C-RAN with massive MIMO is envisioned as an appealing solution, since none of these techniques can solely achieve the 5G targets~\cite{HE15,Mugen_Peng_2014_mag}. In \cite{Mugen_Peng_2014_mag}, the opportunities and challenges for heterogenous C-RAN with massive MIMO were discussed, and it was mentioned that the proper densities of the massive MIMO macrocell BSs (MBSs) and RRHs in the networks should be addressed. While the significance of  heterogenous C-RAN with massive MIMO has been highlighted in the prior works~\cite{mugen_poor_2015,Mugen_Peng_2014_mag}, more research efforts should be devoted to proper characterization of this combination.

Although C-RAN can effectively mitigate the inter-RRH interference by using interference management techniques such as coordinated multi-point (CoMP),  the inter-tier interference between the RRHs and MBSs may be problematic in the heterogeneous C-RAN,  due to the limited radio resources. { Soft fractional frequency reuse (S-FFR) is viewed as an efficient inter-tier interference coordination approach. In \cite{Peng_M_2015}, S-FFR was considered in the heterogeneous C-RAN to both mitigate the inter-tier interference and enhance the spectrum efficiency.}

Recent developments have showed physical layer security as an innovative solution for safeguarding wireless networks. The rationale behind this  is to exploit the randomness inherent in wireless channels such as fading or artificial noise, etc. in order to transmit information confidentially~\cite{LunDong}. In contrast to traditional cryptographic approaches, physical layer security based techniques do not rely on computational complexity and have very good scalability~\cite{Lifeng_commag}. The emergence of massive MIMO also introduces new opportunities for providing physical layer security, e.g.,~\cite{junzhu2014,Lifeng_commag,GanZheng2015}. In particular, in \cite{junzhu2014}, matched filter precoding and artificial noise generation were considered to secure downlink transmission in a multicell massive MIMO system in the presence of an eavesdropper. Subsequently in~\cite{GanZheng2015}, passive eavesdropping and active attacks were investigated in massive MIMO systems with physical layer security, which illustrates that passive eavesdropping has little effect on the secrecy capacity for the case of considering only one single-antenna eavesdropper. While these recent contributions certainly laid a solid foundation in massive MIMO systems with physical layer security, {such a research area} is still far from being well understood. The research on physical layer security in the C-RAN is also in its infancy, and we believe it is a new highly rewarding candidate for physical layer security due to at least the following two crucial factors:
\begin{itemize}
\item Low-power RRHs are densely placed in C-RAN~\cite{checko_2015} so the distance between user and its serving RRH is short, which decreases the risk of information leakage.
\item The inter-RRH interference is mitigated in the C-RAN. As such, all other RRHs can act as ``friendly jammers'' to confound the eavesdroppers~\cite{GanZheng2011,huangjing,Chu2015}.
\end{itemize}
Thus, massive MIMO and C-RAN offer a wealth of opportunities at the physical layer to secure communication.

Motivated by the aforementioned background, in this paper, we explore the benefits of massive MIMO aided heterogeneous C-RAN by investigating its secrecy and EE performance.  We consider downlink transmission in a two-tier heterogeneous C-RAN, in which the RRHs co-exist with the massive MIMO aided macrocells. {To control the inter-tier interference to an acceptable level, S-FFR is used to allocate the radio resources appropriately.} Different from \cite{Peng_M_2015,junzhu2014,GanZheng2015}, in this paper, the RRHs and massive MIMO MBSs are spatially distributed under the framework of stochastic geometry. While \cite{ZhiguoDing_2013,Khan_2015} considered only one single user in the network with multiple RRHs around the user coverage area and evaluated the performance from the standpoint of the user, we analyze the secrecy and EE of the entire network. In summary, our contributions are that:
\begin{itemize}
\item We provide a tractable analytical framework to characterize the secrecy and EE performance of heterogeneous C-RAN aided by massive MIMO. Our analysis accounts for the key features of massive MIMO and C-RAN, i.e., large antenna arrays and simultaneously serving multiple users for massive MIMO, and large numbers of RRHs and inter-RRH interference mitigation for C-RAN.
\item We also study the area ergodic secrecy rate and secrecy outage probability in this network. Our results illustrate that accommodating more users by the massive MIMO empowered MBSs increases the area ergodic secrecy rate and decreases the secrecy outage probability, while it has negligible effect on the RRH's  performance. Deploying large numbers of RRHs increases the area ergodic secrecy rate and decreases the secrecy outage probability.
\item In addition, our results demonstrate that  the effect of S-FFR on the area ergodic secrecy rate of the network can be distinct depending on the RRH density. Moreover, the EE of the RRH tier linearly increases with their dedicated radio resources, and the network EE is improved by using more RRHs and  more radio resources to be allocated to the RRHs.
\end{itemize}

\section{System Descriptions}
\begin{figure}
    \begin{center}
        \includegraphics[width=3.5 in]{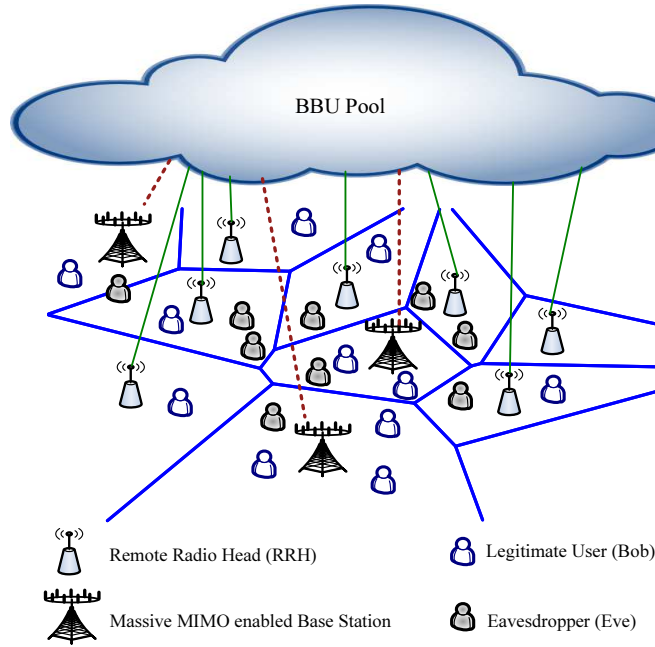}
        \caption{An illustration of two-tier heterogeneous C-RAN, where the red dash lines represent the backhual links between the macrocell base stations and BBU pool via X2/S1 interfaces, and the green solid lines represent the fronthaul links between the RRHs and BBU pool via optical fiber link.}
        \label{HCNs_stochastic}
    \end{center}
\end{figure}
\subsection{Network Model}
As shown in Fig. \ref{HCNs_stochastic}, we consider the downlink of a two-tier heterogeneous C-RAN, where the BBU pool in the cloud is established to coordinate the entire network. Massive MIMO enabled MBSs of the first tier, as high power nodes (HPNs), are connected with the BBU pool via backhaul link, while the RRHs of the second tier, as low power nodes (LPNs), are connected with the BBU pool via fronthaul link (optical fibre link). In this model, we have eavesdroppers (Eves) passively intercepting the secrecy messages without any active attacks.  The locations of Eves are modeled  as a homogeneous Poisson point process (HPPP) $\Phi_e$ with density $\lambda_e$.\footnote{In practice, the behavior of users is unknown and they can also act as malicious Eves, therefore, it is reasonable to assume that the locations of Eves follow PPP~\cite{Geraci_downlink}.}  On the other hand, the locations of MBSs are modeled as an independent HPPP $\Phi_\mathrm{M}$ with density $\lambda_\mathrm{M}$, and we model the locations of RRHs by an independent HPPP $\Phi_\mathrm{R}$ with density $\lambda_\mathrm{R}$.

Equipped with $N_\mathrm{M}$ antennas, each MBS uses zero-forcing beamforming (ZFBF) to communicate with $S$ single-antenna users
over the same resource block (RB) ($N_\mathrm{M}\gg S \geq 1$) using equal power assignment. { The ZFBF matrix at a MBS is ${\bf{W}}{\rm{ = }}{\bf{G}}{\left( {{{\bf{G}}^H}{\bf{G}}} \right)^{ - 1}}$ with the channel matrix $\bf{G}$~\cite{Hosseini2014_Massive}, where $(\cdot)^H$ denotes the Hermitian transpose.} Each RRH is equipped with a single-antenna
and serves a single-antenna user over one RB. {All channels are assumed to undergo independent identically distributed (i.i.d.)
quasi-static Rayleigh block fading.} Further, each user is assumed to be connected with its nearest BS such that the Euclidean plane
is divided into Poisson-Voronoi cells.
%~\footnote{In reality, there may be more than one active users in a small cell and this can be dealt with using multiple access techniques.}.

We consider the adoption of S-FFR for inter-tier interference mitigation and assume that there are a total of $K$ RBs, the number of RBs allocated to RRHs is $\alpha K$, and the number of RBs shared by RRHs and MBSs is $\left(1-\alpha\right) K$, in which $\alpha$ denotes the S-FFR factor, with $\left(0 \leq \alpha \leq 1\right)$. Since inter-RRH interference can be efficiently mitigated via cooperation among RRHs, same radio resources can be shared among the RRHs in the C-RAN~\cite{Peng_M_2015}. For RRH transmission over the $k$-th RB  allocated to RRHs, the receive signal-to-interference-plus-noise ratio (SINR) at a typical user can be expressed as
\begin{align}\label{SINR_RRH_1_k}
\gamma_{\mathrm{R},k} &=\frac{P_\mathrm{R}}{B_oN_o}h_{\mathrm{R},k} \beta \left| {{X_{{o},\mathrm{R}}}} \right|
^{-\eta_\mathrm{R}}, \quad k=1,\dots, \alpha K
\end{align}
where  $P_\mathrm{R}$ is the RRH transmit power allocated to each RB, {$B_o$ is the bandwidth per RB, $h_{\mathrm{R},k} \sim \rm{exp}(1)$ is the small-scale fading channel power gain, $\beta$ is a frequency dependent constant value, which is typically set as ${(\frac{{\text{c}}}{{4\pi {f_c}}})^2}$ with $c=3 \times 10^8 \rm m/s$} and the carrier frequency $f_c$, $\eta_\mathrm{R}$ is the pathloss exponent, $\left| {{X_{{o},\mathrm{R}}}} \right|$ denotes the  distance between the typical user and its  typical serving RRH,  and $N_o$ is the power spectrum density of the noise and the weak inter-RRH interference. For RRH transmission over the $\nu$-th RB  shared by the RRHs and MBSs, the receive SINR at a typical user is written as
\begin{align}\label{SINR_RRH_1_V}
\gamma_{\mathrm{R},\nu}&=\frac{P_\mathrm{R}h_{\mathrm{R},\nu}\beta \left| {{X_{{o},\mathrm{R}}}} \right|
 ^{-\eta_\mathrm{R}}}{I_{\mathrm{M},\nu}+B_oN_o}, \quad \nu=1,\dots, \left(1-\alpha\right) K
\end{align}
where  $h_{\mathrm{R},\nu}\sim \rm{exp}(1)$ is the small-scale fading channel power gain,  $I_{\mathrm{M},\nu}$ is the inter-tier interference from the MBSs, which is given by
 \begin{align}
 I_{\mathrm{M},\nu}=\sum\limits_{\ell  \in {\Phi_\mathrm{M}}} {\frac{P_\mathrm{M}}{S}{h_{\ell,\nu} }\beta \left| {{X_{\ell ,\mathrm{M}}}} \right|^{-\eta_\mathrm{M}}},
 \end{align}
where $P_\mathrm{M}$ is the MBS transmit power of each RB, $h_{\ell,\nu} \sim \Gamma\left(S,1\right)$\footnote{$\Gamma\left(\cdot,\cdot\right)$ is the upper incomplete gamma function~\cite[(8.350)]{gradshteyn}.} is the small-scale fading interfering channel power gain, $\left| {{X_{\ell ,\mathrm{M}}}} \right|$ is the distance between the interfering MBS $\ell \in {\Phi_\mathrm{M}}$ and the typical user, and $\eta_\mathrm{M}$ is the pathloss exponent.

 We consider the non-colluding eavesdropping scenario where the most malicious Eve i.e., the one with the largest SINR of the received signal, dominates the secrecy rate~\cite{LunDong}. Thus, for RRH transmissions over the $k$-th and $\nu$-th RB, the receive SINRs at the most malicious Eve ${e^*}$ are given by
 \begin{align}\label{gamma_e_RRH_k_SINR}
 \gamma_{\mathrm{R},k}^{e^*}=\mathop {\max }\limits_{e \in \Phi_e}\left\{\frac{{{P_\mathrm{R}}h_{\mathrm{R}{\rm{,}}k}^e\beta {{\left| {X_{o{\rm{,}}\mathrm{R}}^e} \right|}^{ - {\eta _\mathrm{R}}}}}}{{I_{\mathrm{R},k}^e + {B_o}{N_e}}}\right\},
 \end{align}
 and
  \begin{align}\label{gamma_e_RRH_1}
 \gamma_{\mathrm{R},\nu}^{e^*}=\mathop {\max }\limits_{e \in \Phi_e}\left\{\frac{{{P_\mathrm{R}}h_{\mathrm{R}{\rm{,}}\nu}^e\beta {{\left| {X_{o{\rm{,}}\mathrm{R}}^e} \right|}^{ - {\eta _\mathrm{R}}}}}}{{I_{\mathrm{R},\nu}^e+I_{\mathrm{M},\nu}^e  + {B_o}{N_e}}}\right\},
 \end{align}
respectively, where $h_{\mathrm{R}{\rm{,}}i}^e (i \in \left\{k,\nu\right\}) \sim \rm{exp}(1)$ and $\left| {X_{o{\rm{,}}\mathrm{R}}^e} \right|$ are the small-scale fading eavesdropping channel power gain and the distance between the typical serving RRH and the Eve $e \in \Phi_e$, respectively, $N_e$ is the noise power spectrum density,   $I_{\mathrm{R},i}^e$ and $I_{\mathrm{M},\nu}^e$ are the interference from the  RRHs and MBSs, which are found as
\begin{equation}\label{interference_M_R_Eve}
\left\{\begin{aligned}
I_{\mathrm{R},i}^e &=\sum\limits_{j  \in {\Phi_\mathrm{R}} /{o}} {{P_\mathrm{R}}{h_{j,i}^e }\beta \left| {{X_{j,\mathrm{R}}^e}} \right|^{-\eta_\mathrm{R}}},\\
I_{\mathrm{M},\nu}^e &=\sum\limits_{\ell  \in {\Phi_\mathrm{M}}} {\frac{P_\mathrm{M}}{S}{h_{\ell,\nu}^e }\beta \left| {{X_{\ell ,\mathrm{M}}^e}} \right|^{-\eta_\mathrm{M}}},
\end{aligned}\right.
\end{equation}
where ${h_{j,i}^e } \sim \rm{exp}(1)$ and $\left| {{X_{j,\mathrm{R}}^e}} \right|$ are the small-scale fading interfering channel power gain and the distance between the RRH $j  \in {\Phi_\mathrm{R}} /{o}$ (except the typical serving RRH) and the Eve, respectively, ${h_{\ell,\nu}^e }\sim \Gamma\left(S,1\right)$~\cite{Hosseini2014_Massive} and $\left| {{X_{\ell ,\mathrm{M}}^e}} \right|$ are the small-scale fading interfering channel power gain and the distance between the MBS $\ell  \in {\Phi_\mathrm{M}} $ and the Eve, respectively.

% In the heterogeneous C-RAN, the inter-MBS interference can be coordinated by the BBU pool through backhaul link~\cite{Mugen_Peng_2014_mag}.
Due to the limited backhaul capacity, the inter-MBS interference is assumed to be not mitigated.
 % \footenote{The case that the inter-MBS interference is not mitigated can be easily considered following \eqref{SINR_RRH_1_V}.}
 Thus, for MBS transmission over the $\nu$-th RB  shared by RRHs and MBSs, the receive SINR at a typical user is written as
 \begin{align}\label{SNR_MBS}
 \gamma_{\mathrm{M},\nu}= \frac{\frac{P_\mathrm{M}}{S}g_{\mathrm{M},\nu}\beta \left| {{X_{{o},\mathrm{M}}}} \right|
^{-\eta_\mathrm{M}}}{J_{\mathrm{M},\nu}+J_{\mathrm{R},\nu}+B_o N_1},
 \end{align}
 where $g_{\mathrm{M},\nu} \sim \Gamma\left(N_\mathrm{M}-S+1,1\right)$ is the small-scale fading channel power gain, $\left| {{X_{{o},\mathrm{M}}}} \right|$ is the distance between the typical user and its typical serving MBS, $N_1$ is the power spectrum density of the noise. In \eqref{SNR_MBS}, $J_{\mathrm{M},\nu}$ and $J_{\mathrm{R},\nu}$ are the interference  from MBSs and RRHs, which are given by
 \begin{equation}\label{interference_M_nu}
 \left\{\begin{aligned}
 J_{\mathrm{M},\nu}&=\sum\limits_{\ell  \in {\Phi_\mathrm{M}}/{o}} {\frac{P_\mathrm{M}}{S}{g_{\ell,\nu}}\beta \left| {{X_{\ell ,\mathrm{M}}}} \right|^{-\eta_\mathrm{M}}},\\
  J_{\mathrm{R},\nu}&=\sum\limits_{j  \in {\Phi_\mathrm{R}}} {{P_\mathrm{R}}{g_{j,\nu} }\beta\left| {{X_{j ,\mathrm{R}}}} \right|^{-\eta_\mathrm{R}}},
 \end{aligned}\right.
\end{equation}
where $g_{\ell,\nu} \sim \Gamma\left(S,1\right)$ and $\left| {{X_{\ell ,\mathrm{M}}}} \right|$ are the small-scale fading interfering channel power gain and the distance between the interfering MBS $\ell  \in {\Phi_\mathrm{M}} /{o}$ (except the typical serving MBS) and the typical user, respectively,   $g_{j,\nu}\sim \exp\left(1\right)$ and  $\left| {{X_{j ,\mathrm{R}}}} \right|$ are the small-scale interfering channel power gain and the distance between the  interfering RBS $j  \in {\Phi_\mathrm{R}}$ and the typical user, respectively.

Similar to \eqref{gamma_e_RRH_1}, for MBS transmission, the receive SINR $\gamma_{\mathrm{M},\nu}^{e^*}$ at the most malicious Eve ${e^*}$ is given by
\begin{align}\label{SINR_M_R_Eve}
\gamma_{\mathrm{M},\nu}^{e^*}=\mathop {\max }\limits_{e \in \Phi_e}\left\{\frac{\frac{P_\mathrm{M}}{S}g_{\mathrm{M},\nu}^{e}\beta \left| {{X_{{o},\mathrm{M}}^{e}}} \right|
^{-\eta_\mathrm{M}}}{J_{\mathrm{M},\nu}^{e}+J_{\mathrm{R},\nu}^{e}+B_o N_e}\right\},
\end{align}
where $g_{\mathrm{M},\nu}^{e} \sim \exp\left(1\right)$ and $\left| {{X_{{o},\mathrm{M}}^{e}}} \right|$ are the small-scale fading channel power gain and distance between the typical serving MBS and the Eve, respectively. In particular, we consider the worst-case scenario that Eves are capable of mitigating the intra-cell interference~\cite{Geraci_downlink}. In \eqref{SINR_M_R_Eve}, $J_{\mathrm{M},\nu}^{e}$ and $J_{\mathrm{R},\nu}^{e}$ are the interference from the MBSs and RRHs, respectively, given by
 \begin{equation}\label{interference_M_R_Eve}
 \left\{\begin{aligned}
J_{\mathrm{M},\nu}^e &=\sum\limits_{\ell  \in {\Phi_\mathrm{M}}/{o}} {\frac{P_\mathrm{M}}{S}{g_{\ell,\nu}^e }\beta \left| {{X_{\ell ,\mathrm{M}}^e}} \right|^{-\eta_\mathrm{M}}},\\
J_{\mathrm{R},\nu}^e &=\sum\limits_{j  \in {\Phi_\mathrm{R}}} {{P_\mathrm{R}}{g_{j,\nu}^e }\beta \left| {{X_{j,\mathrm{R}}^e}} \right|^{-\eta_\mathrm{R}}},
\end{aligned}\right.
\end{equation}
where $g_{\ell,\nu}^e \sim \Gamma\left(S,1\right)$ and $\left| {{X_{\ell ,\mathrm{M}}^e}} \right|$ are the small-scale fading interfering channel power gain and the distance between the interfering MBS $\ell  \in {\Phi_\mathrm{M}} /{o}$ (except the typical serving MBS) and Eve, respectively, and ${g_{j,\nu}^e }\sim  \exp\left(1\right)$ and $\left| {{X_{j,\mathrm{R}}^e}} \right|$ are the small-scale fading interfering channel power gain and the distance between the interfering RRH $j  \in {\Phi_\mathrm{R}}$ and Eve, respectively.

%From  \eqref{SINR_RRH_1_k}, \eqref{SINR_RRH_1_V}, \eqref{gamma_e_RRH_1}, \eqref{SNR_MBS},  and \eqref{SINR_M_R_Eve}, we have the following remark:
%
%{\bf{Remark 1:}} { \emph{For downlink transmission in the heterogeneous C-RAN, legitimate users receive much less interference than the eavesdroppers, owing to the C-RAN's interference management.} }

\subsection{Power Consumption Model}
The total power consumption at each  RRH is given by
\setcounter{equation}{10} \begin{align}\label{RRH_power_cost}
P_\mathrm{R}^{total} = K \frac{{{P_\mathrm{R}}}}{\varepsilon_\mathrm{R}}+P_\mathrm{R}^0+ P_\mathrm{fh} ,
\end{align}
in which $\varepsilon_\mathrm{R}$ is the efficiency of the power amplifier, $P_\mathrm{R}^0$ is the static hardware power consumption of the RRH,  and $P_\mathrm{fh} $ denotes the power consumption of the fronthaul link.

{We employ a general massive MIMO power consumption model proposed in \cite{BE2014}, which can clearly specify how the power scales with the number of antennas and active users in each macrocell. Thus, the total power consumption at each MBS is found as
\begin{align}\label{power_cost_MBS}
P_\mathrm{M}^{total} =  &{\left(1-\alpha\right) K} \Big(\frac{{{P_\mathrm{M}}}}{\varepsilon_\mathrm{M} } + \sum\limits_{\rho = 1}^3 \big( {\left(S\right)^\rho}{\Lambda_{\rho,0}}\nonumber\\
&\;\;\;\;~~~+ {\left(S\right)^{\left(\rho-1\right)}}N_\mathrm{M} {\Lambda_{\rho,1}} \big) \Big)+P_\mathrm{M}^0 + P_\mathrm{bh},
\end{align}
where  $\varepsilon_\mathrm{M} \left(0< \varepsilon_\mathrm{M} \leq 1\right)$ is the efficiency of the power amplifier, {the parameters $\Lambda_{\rho,0}$ and $\Lambda_{\rho,1}$ depend on the transceiver chains, coding and decoding, precoding, etc., which are detailed in  Section \ref{SR_2016}}, $P_\mathrm{M}^0$ is the MBS's static hardware power consumption, and $P_\mathrm{bh}$ is the power consumption of the backhaul link.}

\section{{Secrecy Performance Analysis}}
In this section, the effects of massive MIMO and C-RAN on the secrecy performance are studied in terms of both the area ergodic secrecy rate and secrecy outage probability.

Secrecy outage probability captures the  probability of both reliability and secrecy for one transmission.

\subsection{Area Ergodic Secrecy Rate}
Area ergodic secrecy rate represents the secrecy capacity limitation of the network, which allows us to investigate the impacts of different densities of RRHs and massive MIMO macrocells on the network secrecy performance. We first study the ergodic capacity of the channel between the typical RRH and its served user, which is given as follows.
\begin{theorem}
When using the $k$-th RB allocated to the RRHs, the ergodic capacity $\bar{C}_{\mathrm{R},k}$ of the channel between the typical RRH and its served user is derived as \eqref{RRH_Ergodic_rate_1} (see top of next page). When using the $\nu$-th RB shared by the RRHs and MBSs, the ergodic capacity $\bar{C}_{\mathrm{R},\nu}$ of the channel between the typical RRH and its served user is derived as \eqref{RRH_Ergodic_rate_2}, where $\mathrm{B}_{\left(\cdot\right)}\left[\cdot,\cdot\right]$
is the incomplete beta function~\cite[(8.391)]{gradshteyn}.\footnote{{ Note that the special functions such as incomplete beta function have been included in the commonly-used mathematical softwares such as Mathematic and Matlab, and can be directly calculated.}}
\end{theorem}
\begin{proof}
A detailed proof is provided in Appendix~A.
\end{proof}
\begin{figure*}[!t]
\normalsize
\begin{align}\label{RRH_Ergodic_rate_1}
\bar{C}_{\mathrm{R},k}= \frac{2\pi}{{\ln 2}} \left({{\lambda _\mathrm{R}} + {\lambda _\mathrm{M}}}\right)\int_0^\infty  {{e^{\frac{{{B_o}{N_o}}}{{{P_{\rm{R}}}\beta }}{x^{{\eta _{\rm{R}}}}}}}\Gamma \left( {0,\frac{{{B_o}{N_o}}}{{{P_{\rm{R}}}\beta }}{x^{{\eta _{\rm{R}}}}}} \right)x{e^{ - \pi \left( {{\lambda _\mathrm{R}} + {\lambda _\mathrm{M}}} \right){x^2}}}dx},
\end{align}
\hrulefill% \vspace*{0pt}
\begin{align}\label{RRH_Ergodic_rate_2}
\bar{C}_{\mathrm{R},\nu}= \frac{2\pi}{{\ln 2}} \left({{\lambda _\mathrm{R}} + {\lambda _\mathrm{M}}}\right)\int_0^\infty  {\left[ {\int_0^\infty  {\frac{{\bar{F}}_{{\gamma_{\mathrm{R},\nu}}\mid \left\{\left| {{X_{o,{\rm{R}}}}} \right|=x\right\}}\left(\gamma\right)}{{1 + \gamma }}d\gamma } } \right]x{e^{ - \pi \left( {{\lambda _\mathrm{R}} + {\lambda _\mathrm{M}}} \right){x^2}}}dx},
\end{align}
with
\begin{align}\label{CDF_R_secrecy_rate}
& {\bar{F}}_{{\gamma_{\mathrm{R},\nu}}\mid \left\{\left| {{X_{o,{\rm{R}}}}} \right|=x\right\}}\left(\gamma\right)={e^{ - \frac{{{B_o}{N_o}}}{{{P_{\rm{R}}}\beta }}{x^{{\eta _{\rm{R}}}}}\gamma }} \times \nonumber\\
& \exp \bigg\{   - {\lambda _\mathrm{M}}2\pi \sum\limits_{\mu  = 1}^S
{{
S\choose
\mu
}{{\left( {  \frac{{x^{{\eta _{\rm{R}}}}}\gamma { {P_{\mathrm{M}}}}}{{P_{\rm{R}}S}}} \right)}^\mu }}
 \frac{{\left({\frac{{x^{{\eta _{\rm{R}}}}}\gamma { {P_{\mathrm{M}}}}}{{P_{\rm{R}}S}}}\right)^{ - \mu  + \frac{2}{{{\eta _\mathrm{M}}}}}}}{{{\eta_\mathrm{M}}}} {B_{\left(- {\frac{{x^{{\eta _{\rm{R}}}-\eta _\mathrm{M}}}\gamma { {P_{\mathrm{M}}}}}{{P_{\rm{R}}S}}
} \right)}}\left[
{\mu  - \frac{2}{{{\eta_\mathrm{M}}}},1 - S} \right]\bigg\}
\end{align}
\hrulefill% \vspace*{0pt}
\end{figure*}
\begin{figure*}[!t]
\normalsize
\begin{align}\label{Eve_Ergodic_capacity_overline_k}
&\bar{C}_{\mathrm{R},k}^{e^*}=\frac{1}{{\ln 2}}\int_0^\infty \frac{1-F_{\gamma_{\mathrm{R},k}^{e^*}}\left(x\right)}{1+x} dx,
\end{align}
with
\begin{align}\label{X3_34_overline_F_k}
& F_{\gamma_{\mathrm{R},k}^{e^*}}\left(x\right)=\exp\bigg\{-2\pi{\lambda _e}\int_0^\infty  \exp\Big[-\frac{{r}^{ {\eta _\mathrm{R}}}x}{{P_\mathrm{R}}\beta}{B_o}{N_e} -\lambda_\mathrm{R}\pi \Gamma\left(1+\frac{2}{\eta _\mathrm{R}}\right) \Gamma\left(1-\frac{2}{\eta _\mathrm{R}}\right) \big({{r}^{ {\eta _\mathrm{R}}}x}\big) ^{\frac{2}{\eta _\mathrm{R}}} \Big] r dr \bigg\}
\end{align}
\hrulefill% \vspace*{0pt}
\begin{align}\label{Eve_Ergodic_capacity_overline}
&\bar{C}_{\mathrm{R},\nu}^{e^*}=\frac{1}{{\ln 2}}\int_0^\infty \frac{1-F_{\gamma_{\mathrm{R},\nu}^{e^*}}\left(x\right)}{1+x} dx,
\end{align}
with
\begin{align}\label{X3_34_overline_F}
 F_{\gamma_{\mathrm{R},\nu}^{e^*}}\left(x\right)=&\exp\bigg\{-2\pi{\lambda _e}\int_0^\infty  \exp\Big[-\frac{{r}^{ {\eta _\mathrm{R}}}x}{{P_\mathrm{R}}\beta}{B_o}{N_e} -\lambda_\mathrm{R}\pi \Gamma\left(1+\frac{2}{\eta _\mathrm{R}}\right) \Gamma\left(1-\frac{2}{\eta _\mathrm{R}}\right) \big({{r}^{ {\eta _\mathrm{R}}}x}\big) ^{\frac{2}{\eta _\mathrm{R}}} \nonumber\\
&-2\pi \lambda_\mathrm{M} \sum\limits_{\mu  = 1}^S
{S\choose \mu} \left(\frac{{r}^{ {\eta _\mathrm{R}}}x P_\mathrm{M}}{P_\mathrm{R} S}\right)^{\frac{2}{\eta_\mathrm{M}}}\frac{\Gamma\left(\mu -\frac{2}{\eta_\mathrm{M}}\right) \Gamma\left(-\mu +\frac{2}{\eta_\mathrm{M}}+S\right)}{\eta_\mathrm{M}\Gamma\left(S\right)} \Big] r dr \bigg\}
\end{align}
\hrulefill% \vspace*{0pt}
\end{figure*}
We next derive the ergodic capacity of the channel between the most malicious eavesdropper and the typical RRH, which is given as follows:
\begin{theorem}
For RRH transmissions over the $k$-th RB  and $\nu$-th RB, the ergodic capacity $\bar{C}_{\mathrm{R},k}^{e^*} $ and $\bar{C}_{\mathrm{R},\nu}^{e^*} $ of the most malicious eavesdropper's channel are derived as \eqref{Eve_Ergodic_capacity_overline_k} and \eqref{Eve_Ergodic_capacity_overline}, respectively, in the next page.

\end{theorem}
\begin{proof}
A detailed proof is provided in Appendix~B.
\end{proof}

Based on \textbf{Theorem 1} and \textbf{Theorem 2 }, using Jensen's inequality that $\mathbb{E}\left\{ {\max \left( {X{\rm{,}}Y} \right)} \right\} \ge \max \left( {\mathbb{E}\left\{ x \right\}{\rm{,}}\mathbb{E}\left\{ Y \right\}} \right)$, the ergodic secrecy rate for the typical RRH transmission over the $k$-th RB is lower bounded as~\cite{Zhou2010,junzhu2014}
\setcounter{equation}{19}\begin{align}\label{RRH_ESR_1}
R_{\mathrm{R},k}^{\mathrm{s}}=\left[\bar{C}_{\mathrm{R},k}-\bar{C}_{\mathrm{R},k}^{e^*}\right]^{+},
\end{align}
where $[x]^+=\max\left\{x,0\right\}$.

Likewise, the ergodic secrecy rate for the typical RRH transmission over the $\nu$-th RB is lower bounded as
\begin{align}\label{RRH_ESR_2}
R_{\mathrm{R},\nu}^{\mathrm{s}}=\left[\bar{C}_{\mathrm{R},\nu}-\bar{C}_{\mathrm{R},\nu}^{e^*}\right]^{+}.
\end{align}
\textbf{Remark 2:} \emph{From the results in \textbf{Theorem 1}, \textbf{Theorem 2}, \eqref{RRH_ESR_1} and \eqref{RRH_ESR_2}, we realize that the ergodic secrecy rate for RRH transmission increases with the density of RRHs, which can be explained by the facts that:  1) When deploying more RRHs in the same area, the distance between the legitimate user and its associated RRH is shorter, which decreases the pathloss; and 2) more interference from RRHs is present at the eavesdroppers, which degrades the eavesdropping channel. }

The area ergodic secrecy rate (in bits/s/m$^2$) of the RRH tier in the heterogeneous C-RAN is calculated as
\begin{align}\label{RRH_tier}
R_{\mathrm{R}}^{\mathrm{s}}=\lambda_{\mathrm{R}}\left(\alpha K B_o R_{\mathrm{R},k}^{\mathrm{s}}+\left(1-\alpha\right) K B_o R_{\mathrm{R},\nu}^{\mathrm{s}}\right).
\end{align}

For MBS transmission, we have a tractable lower bound expression for the ergodic capacity of the channel between the typical MBS and its serving user as stated in the following theorem.
\begin{theorem}
For MBS transmission over the $\nu$-th RB, the ergodic capacity  of the channel between the typical MBS and its served user is lower bounded in closed-form as \begin{align}\label{MBS_EC_nu}
&\vspace{-0.5 cm}\bar{C}_{\mathrm{M},\nu}^{\mathrm{L}}=\log_2\Bigg(1+\exp\bigg(\ln \left( {\frac{{{P_{\rm{M}}}}}{S}\beta } \right)+ \psi\left(N_\mathrm{M}-S+1\right)-\frac{\eta _{\rm{M}}}{2}\nonumber\\
 &\vspace{-0.5 cm}\left( \psi \left( 1 \right)
 - \ln \left( {\pi \left( {{\lambda _{\rm{R}}} + {\lambda _{\rm{M}}}} \right)} \right) \right)-\ln\Big( {\frac{{{P_\mathrm{M}}\beta 2\pi {\lambda _\mathrm{M}}} \Gamma\left(2-\frac{\eta_\mathrm{M}}{2}\right) }{{\left({\eta_\mathrm{M}} - 2\right)\left(\pi\lambda _\mathrm{M}+\pi\lambda _\mathrm{R}\right)^{1-\frac{\eta_\mathrm{M}}{2}}}}}\nonumber\\
& +{\frac{{{P_\mathrm{R}}\beta 2\pi {\lambda _\mathrm{R}}} \Gamma\left(2-\frac{\eta_\mathrm{R}}{2}\right) }{{\left({\eta_\mathrm{R}} - 2\right)\left(\pi\lambda _\mathrm{M}+\pi\lambda _\mathrm{R}\right)^{1-\frac{\eta_\mathrm{R}}{2}}}}}+ B_o N_1\Big)\bigg)\Bigg),
\end{align}
where $\psi\left(\cdot\right)$ is the digamma function~\cite{Abramowitz}. For very large $N_\mathrm{M}$, $\psi\left(N_\mathrm{M}-S+1\right)\approx \ln\left(N_\mathrm{M}-S+1\right)$~\cite{Lifeng_massiveMIMO}.
\end{theorem}

\begin{proof}
A detailed proof is provided in Appendix~C.
\end{proof}

For MBS transmission over the $\nu$-th RB, the ergodic capacity $\bar{C}_{\mathrm{M},\nu}^{e^*} $ of the most malicious eavesdropper's channel is derived as
\begin{align}\label{Eve_MBS_rate}
\bar{C}_{\mathrm{M},\nu}^{e^*}&=\mathbb{E}\left\{\log_2\left(1+\gamma_{\mathrm{M},\nu}^{e^*}\right)\right\}\nonumber\\
&=\frac{1}{{\ln 2}} \int_0^\infty \frac{1-F_{\gamma_{\mathrm{M},\nu}^{e^*}}\left(x\right)}{1+x} dx,
\end{align}
where $F_{\gamma_{\mathrm{M},\nu}^{e^*}}\left(x\right)$ is given by \eqref{F_gamma_Me} in the next page, which can be easily obtained by following the proof of \textbf{Theorem 2}.
\begin{figure*}
\begin{align}\label{F_gamma_Me}
&F_{\gamma_{\mathrm{M},\nu}^{e^*}}\left(x\right)= \exp\bigg\{-2\pi{\lambda _e}\int_0^\infty  \exp\bigg[-\frac{S{r}^{ {\eta _\mathrm{M}}} x}{{P_\mathrm{M}}\beta}{B_0}{N_1} -\lambda_\mathrm{R}\pi \left(P_\mathrm{R}\beta\right)^{\frac{2}{\eta _\mathrm{R}}}\Gamma\left(1+\frac{2}{\eta _\mathrm{R}}\right) \Gamma\left(1-\frac{2}{\eta _\mathrm{R}}\right)  \nonumber\\
&\left(\frac{S{r}^{ {\eta _\mathrm{M}}}x}{{P_\mathrm{M}}\beta}\right) ^{\frac{2}{\eta _\mathrm{R}}}-2\pi \lambda_\mathrm{M} \sum\limits_{\mu  = 1}^S
{S\choose \mu} \left({{r}^{ {\eta _\mathrm{M}}}x }\right)^{\frac{2}{\eta_\mathrm{M}}}\frac{\Gamma\left(\mu -\frac{2}{\eta_\mathrm{M}}\right) \Gamma\left(-\mu +\frac{2}{\eta_\mathrm{M}}+S\right)}{\eta_\mathrm{M}\Gamma\left(S\right)} \Bigg] r dr \Bigg\}
\end{align}
\hrulefill% \vspace*{0pt}
\end{figure*}
Based on \textbf{Theorem 3} and \eqref{Eve_MBS_rate}, the ergodic secrecy rate for the typical MBS transmission over the $\nu$-th RB is lower bounded as
\setcounter{equation}{25}\begin{align}\label{MBS_ESR_1}
R_{\mathrm{M},\nu}^{\mathrm{s}}=\left[\bar{C}_{\mathrm{M},\nu}^{\mathrm{L}}-\bar{C}_{\mathrm{M},\nu}^{e^*}\right]^{+}.
\end{align}

\textbf{Remark 3:} \emph{From the results in \eqref{MBS_EC_nu}, \eqref{Eve_MBS_rate}, and \eqref{MBS_ESR_1}, we establish that the ergodic secrecy rate is improved by increasing the number of MBS antennas, due to the fact that only the served legitimate users can obtain the large array gains.}

The area ergodic secrecy rate (in bits/s/m$^2$) of the MBS tier in the heterogeneous C-RAN is calculated as
\begin{align}\label{MBS_tier}
R_{\mathrm{M}}^{\mathrm{s}}=\lambda_{\mathrm{M}}\left(1-\alpha\right) K B_o S R_{\mathrm{M},\nu}^{\mathrm{s}}.
\end{align}

\subsection{{Secrecy Outage Probability}}
In the above, we have studied the secrecy capacity in the massive MIMO aided heterogeneous C-RAN. Since Eves only intercept the secrecy massages passively without any transmissions, the channel state information (CSI) of the eavesdropping channels cannot be obtained by the BSs or legitimate users. In this circumstance, the BSs set the transmission rate $R$ consisting of the secrecy codewords and non-secrecy codewords, and  a constant rate of the secrecy codewords $R_\mathrm{s }\left(R_\mathrm{s} \leq R \right)$. Secrecy outage is declared when the targeted secrecy rate $R_\mathrm{s}$ cannot be guaranteed.

\subsubsection{Delay-Limited Mode}
In the delay-limited mode, a rate $R$ is set under certain connection outage constraint. For RRH transmission over the $k$-th RB  allocated to the RRHs, given a distance $\left| {{X_{{0},\mathrm{R}}}} \right|=d_0$ between a typical RRH and its serving user, the connection outage probability is given by
\begin{align}\label{RRH_COP_1}
P_{\mathrm{R},k}^{\mathrm{co}}\left(R\right)&=\Pr\left(\log_2\left(1+\gamma_{\mathrm{R},k}\right)<R\right)\nonumber\\
&=1-\exp\left(-\frac{B_oN_o}{P_\mathrm{R}\beta}d_o^{\eta_\mathrm{R}}\left(2^{R}-1\right) \right).
\end{align}
For RRH transmission over the $\nu$-th RB shared by RRHs and MBSs, the connection outage probability is
\begin{align}\label{RRH_COP_2}
P_{\mathrm{R},\nu}^{\mathrm{co}}\left(R\right)&=\Pr\left(\log_2\left(1+\gamma_{\mathrm{R},\nu}\right)<R\right)\nonumber\\
&=1-{\bar{F}}_{{\gamma_{\mathrm{R},\nu}}\mid \left\{\left| {{X_{o,{\rm{R}}}}} \right|=d_o\right\}}\left(2^R-1\right),
\end{align}
where ${\bar{F}}_{{\gamma_{\mathrm{R},\nu}}\mid \left\{\left| {{X_{o,{\rm{R}}}}} \right|=x\right\}}\left(\cdot\right)$ is given by \eqref{CDF_R_secrecy_rate}.

\textbf{Remark 4:} \emph{From \eqref{RRH_COP_1} and \eqref{RRH_COP_2}, we see that when a typical RRH transmits information to its served user, deploying more RRHs in its surrounding area has no effect on the quality of connectivity, since the inter-RRH interference is mitigated. }

\begin{corollary}
Given a distance  $\left| {{X_{{o},\mathrm{R}}}} \right|=d_o$ and the connection outage probability threshold $\sigma$, the typical RRH transmission rate over the $k$-th RB allocated to RRHs is given by
\setcounter{equation}{29}\begin{align}\label{DL_RRH_1}
R= {\log _2}\left( {1 - \frac{{{P_{\rm{R}}}\beta }}{{{B_o}{N_o}{d_o^{{\eta _{\rm{R}}}}}}}\ln \left( {1 - \sigma } \right)} \right),
\end{align}
and the typical RRH transmission rate over the $\nu$-th RB shared by RRHs and MBSs satisfies
\begin{align}\label{coro_RRH_rate2}
R \geq
\log_2\left(1+\frac{{{P_\mathrm{R}}S}}{{{P_\mathrm{M}}{d_o^{{\eta _R}}}}} \Delta_1^{\frac{{{\eta _\mathrm{M}}}}{2}} \right),
\end{align}
with
\begin{align}\label{coro_RRH_rate2_gamma}
\Delta_1=
 { - \frac{{{\eta _\mathrm{M}}\Gamma \left( S \right)\ln \left( {1 - \sigma } \right)}}{{2\pi {\lambda_\mathrm{M}}\sum\limits_{\mu  = 1}^S {{S\choose
\mu}
\Gamma \left( {\mu  - \frac{2}{{{\eta _\mathrm{M}}}}} \right)\Gamma \left( { - \mu  + \frac{2}{{{\eta _\mathrm{M}}}} + S} \right)} }}}.
\end{align}

\end{corollary}
\begin{proof}
A detailed proof is provided in Appendix D.
\end{proof}

Similar to \eqref{RRH_COP_2}, given a distance  $\left| {{X_{{o},\mathrm{M}}}} \right|=d_o$ between a typical MBS and its served user, we obtain the connection outage probability of MBS transmission as
\begin{align}\label{MBS_COP}
 \mathrm{P}_{\mathrm{M},\nu}^{\mathrm{co}}\left(R\right)= 1-{\bar{F}}_{{\gamma_{\mathrm{M},\nu}}\mid \left\{\left| {{X_{o,{\rm{M}}}}} \right|=d_o\right\}}\left(2^{R}-1\right),
\end{align}
where  ${\bar{F}}_{{\gamma_{\mathrm{M},\nu}}\mid \left\{\left| {{X_{o,{\rm{M}}}}} \right|=d_o\right\}}\left(\cdot\right)$ is  the complementary cumulative distribution function (CCDF) of the receive SINR $\gamma_{\mathrm{M},\nu}$ at the MBS. However,  the exact expression for ${\bar{F}}_{{\gamma_{\mathrm{M},\nu}}\mid \left\{\left| {{X_{o,{\rm{M}}}}} \right|=d_o\right\}}\left(\cdot\right)$ involves higher order derivatives of laplace transform using $\mathrm{Fa\grave{a}}$
di Bruno's formula~\cite{Dhillon2013_HetNets}, which becomes inefficient for large number of MBS antennas. By the law of large numbers, i.e., $ g_{\mathrm{M},\nu} \approx N_\mathrm{M}-S+1$ as $N_\mathrm{M}$ is large, and employing Gil-Pelaez theorem~\cite{Renzo_comml_2014}, we have
\begin{align}\label{CCDF_MBS}
&{\bar{F}}_{{\gamma_{\mathrm{M},\nu}}\mid \left\{\left| {{X_{o,{\rm{M}}}}} \right|=d_o\right\}}\left(\gamma\right)= \Pr\Big(\frac{\frac{P_\mathrm{M}}{S}\left(N_\mathrm{M}-S+1\right)\beta d_o^{-\eta_\mathrm{M}}}{J_{\mathrm{M},\nu}+J_{\mathrm{R},\nu}+B_o N_1} > \gamma \Big)\nonumber\\
&=\Pr\Big(J_{\mathrm{M},\nu}+J_{\mathrm{R},\nu}< \Big(\frac{P_\mathrm{M}}{S \gamma}\left(N_\mathrm{M}-S+1\right)\beta d_o^{-\eta_\mathrm{M}}-B_o N_1\Big)\Big)\nonumber\\
&= \frac{1}{2}-\frac{1}{\pi}\int_0^\infty \frac{\mathrm{Im}\big[e^{-jw\big(\frac{P_\mathrm{M}\beta}{Sd_o^{\eta_\mathrm{M}} \gamma}\left(N_\mathrm{M}-S+1\right) -B_o N_1\big)}\varphi^{*}\left(w\right)\big]}{w}dw,
\end{align}
where  $j=\sqrt{-1}$, $\varphi\left(w\right)$ is the conjugate of the characteristic function given by \eqref{Varphi_w} (See top of this page), which can be easily obtained by following the similar approach in Appendix A. In \eqref{Varphi_w},  $_2{F_1}\left[\cdot,\cdot;\cdot;\cdot\right]$ is the Gauss hypergeometric function~\cite[(9.142)]{gradshteyn}.
\begin{figure*}
\normalsize
\begin{align}\label{Varphi_w}
&\varphi\left(w\right)=\exp \bigg( -  2 \pi {\lambda_\mathrm{R}} \frac{jw P_\mathrm{R} \beta d_o^{2-\eta_\mathrm{R}}}{\eta_\mathrm{R}-2} {}_2{F_1}\left[ 1,{\frac{{\eta_\mathrm{R}- 2}}{{{\eta_\mathrm{R}}}};2 - \frac{2}{{{\eta_\mathrm{R}}}}; - jw P_\mathrm{R}\beta {d_o^{ - {\eta_\mathrm{R}}}}} \right]- {\lambda _\mathrm{M}}2\pi \sum\limits_{\mu  = 1}^S
 \nonumber\\
&{{
S\choose
\mu
}{{\left( {jw \frac{{ {P_{\mathrm{M}}}}}{{S}}}\beta \right)}^\mu }}  \frac{{{{\left( { - jw \frac{{
{P_{\mathrm{M}}}}}{{S}}}  \beta \right)}^{ - \mu  + \frac{2}{{{\eta _\mathrm{M}}}}}}}}{{{\eta_\mathrm{M}}}}{B_{\left(- {jw \frac{{ {P_{\mathrm{M}}}}}{{S}} \beta
d_o^{-\eta _\mathrm{M}}} \right)}}\left[
{\mu  - \frac{2}{{{\eta_\mathrm{M}}}},1 - S} \right]\bigg)
\end{align}
\hrulefill% \vspace*{0pt}
\end{figure*}

%\begin{align}\label{CPB_gpt_OP_2}
%&\varphi\left(w\right)=\mathbb{E}\left\{e^{-jw  J_{\mathrm{R},\nu} }\right\} \nonumber\\
%&\mathop= \limits^{\left( a \right)}\exp\left(-2\pi \lambda_\mathrm{R}\int_{d_o}^\infty \left(1-\frac{1}{\left(1+jw P_\mathrm{R} \beta r^{-\eta_\mathrm{R}} \right)}\right)r dr    \right)\nonumber\\\nonumber\\
%&= \exp \bigg( -  2 \pi {\lambda_\mathrm{R}} \frac{jw P_\mathrm{R} \beta d^{2-\eta_\mathrm{R}}}{\eta_\mathrm{R}-2} \times \nonumber\\
%&\quad\quad\quad\quad\quad\quad {}_2{F_1}\left[ 1,{\frac{{\eta_\mathrm{R}- 2}}{{{\eta_\mathrm{R}}}};2 - \frac{2}{{{\eta_\mathrm{R}}}}; - jw P_\mathrm{R}\beta {d_o^{ - {\eta_\mathrm{R}}}}} \right]  \bigg),
%\end{align}
%where (a) results from the generating functional of the PPP $\Phi_e$, $_2{F_1}\left[\cdot,\cdot;\cdot;\cdot\right]$ is the Gauss hypergeometric function~\cite[(9.142)]{gradshteyn}.

Secrecy outage occurs when the equivocation rate of Eve is lower than the secrecy rate $R_\mathrm{s}$. Thus, the secrecy outage probability can be written in a general form as
\setcounter{equation}{35}\begin{align}\label{RRH_exp_SOP_DL}
P_\mathrm{s}&=\Pr\left(R-\log_2\left(1+ \gamma_{\vartheta,i}^{e^*} \right)<R_\mathrm{s}\right)\nonumber\\
&=1- F_{\gamma_{\vartheta,i}^{e^*}}\left(2^{R-R_\mathrm{s}}-1\right),
\end{align}
where $F_{\gamma_{\vartheta,i}^{e^*}} \left(\cdot\right)$  is the CDF of the SINR $\gamma_{\vartheta,i}^{e^*}$ at the most malicious Eve. Note that here,
{ $F_{\gamma_{\vartheta,i}^{e^*}} \left(\cdot\right)$ is given by \eqref{X3_34_overline_F_k}, \eqref{X3_34_overline_F} and \eqref{F_gamma_Me}
 for RRH transmissions ($\vartheta={\rm{R}}, i \in \left\{k,\nu\right\}$)  and MBS transmissions ($\vartheta={\rm{M}}, i =\nu$), respectively.}  %$\epsilon_\mathrm{DT}$ is the threshold.

\textbf{Remark 5:} \emph{From \eqref{RRH_exp_SOP_DL}, we see that for eavesdroppers, deploying more RRHs and MBSs results in more interference, which degrades the eavesdropping channel, thereby decreasing the secrecy outage probability.}

\subsubsection{Delay-Tolerant Mode}
In the delay-tolerant mode, coding can be operated over a sufficient number of independent channel realizations to experience the whole ensemble of the channel, and therefore the transmission rate $R$ can be set as an arbitrary value less than or equal to the ergodic capacity of the channel between the legitimate user and its serving RRH/MBS~\cite{junzhu2014}. The secrecy outage occurs when the targeted ergodic secrecy rate $R_\mathrm{s}$ cannot be satisfied, i.e.,
\begin{align}\label{RRH_exp_SOP}
R-R_e<R_\mathrm{s},
\end{align}
where $R_e$ denotes the ergodic capacity of the most malicious eavesdropper's channel. When intercepting the RRH transmission, $R_e=\bar{C}_{\mathrm{R},i}^{e^*} (i \in \left\{k, \nu\right\})$ given by \eqref{Eve_Ergodic_capacity_overline_k} and \eqref{Eve_Ergodic_capacity_overline} respectively; and when intercepting the MBS transmission, $R_e=\bar{C}_{\mathrm{M},\nu}^{e^*}$ given by \eqref{Eve_MBS_rate}. As mentioned in \textbf{Remark 5}, $R_e$ decreases with increasing the densities of RRHs and MBSs, due to more severe interference.

It is indicated from \eqref{RRH_exp_SOP} that given a secrecy rate $R_\mathrm{s}$, the rate $R$ should be set as large as possible to avoid the secrecy outage. Based on {\bf{Theorem 1}}, i) RRH transmission over the $k$-th RB  allocated to RRHs, the transmission rate $R$ at RRH (bits/s/Hz) satisfies $R\leq \bar{C}_{\mathrm{R},k}$ with $\bar{C}_{\mathrm{R},k}$ given by \eqref{RRH_Ergodic_rate_1}, and ii) RRH transmission over the $\nu$-th RB shared by the RRHs and MBSs, $R$ at RRH satisfies $R \leq \bar{C}_{\mathrm{R},\nu}$ with $\bar{C}_{\mathrm{R},\nu}$ given by \eqref{RRH_Ergodic_rate_2}. Based on {\bf{Theorem 3}}, for MBS transmission over the $\nu$-th RB shared by the RRHs and MBSs, The value for $R$ at MBS can be at least set as $R = \bar{C}_{\mathrm{M},\nu}^{\mathrm{L}}$ with $\bar{C}_{\mathrm{M},\nu}^{\mathrm{L}}$ given by \eqref{MBS_EC_nu}.

\section{{Energy Efficiency Analysis}}
One of the 5G goals is to achieve $10{\rm x}$ reduction in energy consumption~\cite{5GPPP_2015}. As such, EE is a very important performance metric. In this section, we proceed to examine the EE concern in the massive MIMO aided heterogeneous C-RAN.\footnote{Note that because the CSI of the eavesdropping channels is unknown, joint design of combining both EE and secrecy is not feasible.}

The EE for transmission from a typical RRH  is given by
\begin{align}\label{EE_RRH}
\mathrm{EE}_{\mathrm{RRH}}&=\frac{\mathrm{throughput}}{\mathrm{Power~Consumption}}\nonumber\\
&=\frac{\alpha K B_o \bar{C}_{\mathrm{R},k}+\left(1-\alpha\right) K B_o \bar{C}_{\mathrm{R},\nu}}{P_\mathrm{R}^{total}},
\end{align}
where $\bar{C}_{\mathrm{R},k}$ and $\bar{C}_{\mathrm{R},\nu}$ are given by \eqref{RRH_Ergodic_rate_1} and \eqref{RRH_Ergodic_rate_2}, respectively, based on {\bf{Theorem 1}}.
In the RRH tier, transmission over RBs that are only allocated to RRHs plays a dominant role in the overall throughput~\cite{Peng_M_2015},
compared to using RBs shared by the RRHs and MBSs. As a consequence, \eqref{EE_RRH}
can be approximately evaluated as % \approx \frac{\alpha K B_o \bar{C}_{\mathrm{R},k}}{P_\mathrm{R}^{total}}
\begin{align}\label{EE_RRH_Approx}
\mathrm{EE}_{\mathrm{RRH}}&\mathop \approx \limits^{\left(a\right)} \frac{\alpha B_o \bar{C}_{\mathrm{R},k}}{ \frac{{{P_\mathrm{R}}}}{\varepsilon_\mathrm{R}}},
\end{align}
where (a) is obtained by omitting the power consumptions from static hardware and fronthaul link, compared to the RRH transmit power. It is implied from \eqref{EE_RRH_Approx} that the EE for RRH transmission can be linearly improved by allocating more RBs to the RRHs. From \eqref{RRH_Ergodic_rate_1}, we note that $\bar{C}_{\mathrm{R},k}$ increases with density of RRHs and MBSs. Hence we have the following corollary:
\begin{corollary}\label{col_2}
EE for RRH transmission is improved by increasing the density of RRHs and MBSs in the heterogeneous C-RAN, due to the fact that the distance between the user and its associated RRH is shorter, hence increasing the throughput.
\end{corollary}

Likewise, the EE for transmission from a typical MBS can at least achieve
\begin{align}\label{EE_MBS}
\mathrm{EE}_{\mathrm{MBS}}&=\frac{\left(1-\alpha\right) K B_o S \bar{C}_{\mathrm{M},\nu}^{\mathrm{L}} }{P_\mathrm{M}^{total}}\nonumber\\
&\mathop \approx \limits^{\left(b\right)} \frac{ B_o S \bar{C}_{\mathrm{M},\nu}^{\mathrm{L}} }{\frac{{{P_\mathrm{M}}}}{\varepsilon_\mathrm{M} } + \sum\limits_{\rho = 1}^3 {\left( {{\left(S\right)^\rho}{\Lambda_{\rho,0}} + {\left(S\right)^{\left(\rho-1\right)}}N_\mathrm{M} {\Lambda_{\rho,1}}} \right)} },
\end{align}
where { $P_\mathrm{M}^{total}$ represents the total power consumption at each MBS given by \eqref{power_cost_MBS}}, (b) is obtained by the fact that the power consumptions from static hardware and backhaul link are negligible compared to the massive MIMO processing. In \eqref{EE_MBS}, $\bar{C}_{\mathrm{M},\nu}^{\mathrm{L}}$  is given by \eqref{MBS_EC_nu}, based on {\bf{Theorem 3}}. From \eqref{EE_MBS}, we see that S-FFR has negligible effect on the EE of MBS transmission.

In light of the aforementioned, we conclude that the EE of the massive MIMO enabled heterogeneous C-RAN is improved by increasing the RRH density and RBs only used by RRHs.

%Moreover, we have the following corollary:
%\begin{corollary}
%Energy efficiency of MBS transmission increases with the density of  MBSs in the heterogeneous C-RAN, however, it decreases with increasing the density of RRHs when ${\eta _\mathrm{M}}<{\eta _\mathrm{R}}$.
%\end{corollary}
%\begin{proof}
%A detailed proof is provided in Appendix E.
%\end{proof}

We next evaluate the EE of this network. Using the \textbf{Theorem 1} and \textbf{Theorem 3} in Section III, we know that the EE of the massive MIMO enabled heterogeneous C-RAN can at lease achieve
\begin{align}\label{EE_network}
&\mathrm{EE}_{\mathrm{Net}}=\frac{\mathrm{Area~throughput~of~the~network}}{\mathrm{Area~Power~Consumption~of~the~network}}\nonumber\\
&=\frac{\lambda _\mathrm{R}\alpha K B_o \bar{C}_{\mathrm{R},k}+\left(1-\alpha\right) K B_o \left( \lambda _\mathrm{R} \bar{C}_{\mathrm{R},\nu} +\lambda _\mathrm{M} S \bar{C}_{\mathrm{M},\nu}^{\mathrm{L}}\right) }{\lambda _\mathrm{R} P_\mathrm{R}^{total}+\lambda _\mathrm{M}P_\mathrm{M}^{total}}.
\end{align}
\section{Numerical Results}\label{SR_2016}
In this section, we present numerical results to evaluate the secrecy and EE of the massive MIMO enabled heterogeneous C-RAN (abbreviated as Het C-RAN in the figures). We consider a circular region with radius $1\times 10^4$ m.   Such a network is assumed to operate at a carrier frequency of 1 GHz, with the MBS transmit power $P_\mathrm{M} = 40$ dBm, the RRH transmit power $P_\mathrm{R}= 30$ dBm, the RB bandwidth $\mathrm{B}_0=800$ kHz, and the total number of RBs $K=25$. The  power spectrum densities are $N_0=N_1=N_e=-162$ dBm/Hz~\cite{mugen_poor_2015}. The static hardware power consumption  for RRH and HPN are $P_\mathrm{R}^0=$0.1 W and $P_\mathrm{M}^0=$10 W, respectively, and the power consumption of the fronthaul link and backhaul link are $P_\mathrm{fh}=P_\mathrm{bh}=0.2$ W. We set the coefficients for power consumption under ZFBF precoding in \eqref{power_cost_MBS} as $P_\mathrm{M}^0=4$ W, $\Lambda_{1,0}=4.8$, $\Lambda_{2,0}=0$, $\Lambda_{3,0}=2.08 \times 10^{-8}$, $\Lambda_{1,1}=1$, $\Lambda_{2,1}=9.5 \times 10^{-8}$ and $\Lambda_{3,1}=6.25 \times 10^{-8}$ \cite{BE2014}. {In the simulation results, the values of MBS and RRH density are set based on the macro inter-site distance (ISD) in 3GPP model~\cite{3GPP_model_2010}}.

\subsection{The Effects of Massive MIMO}
\begin{figure}
    \begin{center}
        \includegraphics[width=3.1 in]{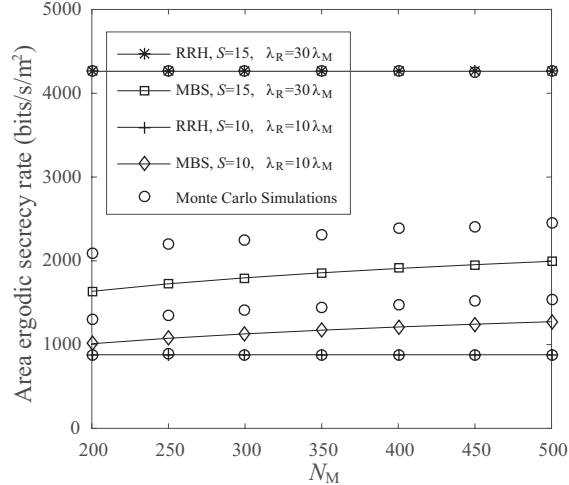}
        \caption{Effects of massive MIMO on the area ergodic secrecy rate: $\lambda_\mathrm{M}=\left(500^2 \times \pi\right)^{-1} \mathrm{m}^{-2}$, $\lambda_e=10^{-5}$ m$^{-2}$, $\eta_\mathrm{M}=3.0$, $\eta_\mathrm{R} = 3.6$, and $\alpha=0.5$.}
        \label{Fig1}
    \end{center}
\end{figure}

Fig.~\ref{Fig1} analyzes the effects of massive MIMO on the area ergodic secrecy rate. The analytical curves for area ergodic secrecy rate of the  RRH tier were obtained from \eqref{RRH_tier}, which have a precise match to the results obtained using the Monte-Carlo simulations marked by `$\circ$'. The lower bound curves for area ergodic secrecy rate of the MBS tier were obtained from using \eqref{MBS_tier}, which can efficiently predict the performance behavior. As mentioned in \textbf{Remark 3} of Section III-A, we observe that the area ergodic secrecy rate increases with the number of MBS antennas, due to more array gains obtained by the legitimate user. Increasing the number of served users can also significantly improve the ergodic secrecy rate. The area ergodic secrecy rate of the  RRH tier remains unchanged with increasing the number of MBS antennas, since employing more MBS antennas will not cause more interference in the network. Nevertheless, it will substantially increase with the density of RRHs.
\begin{figure}
    \begin{center}
        \includegraphics[width=3.1 in]{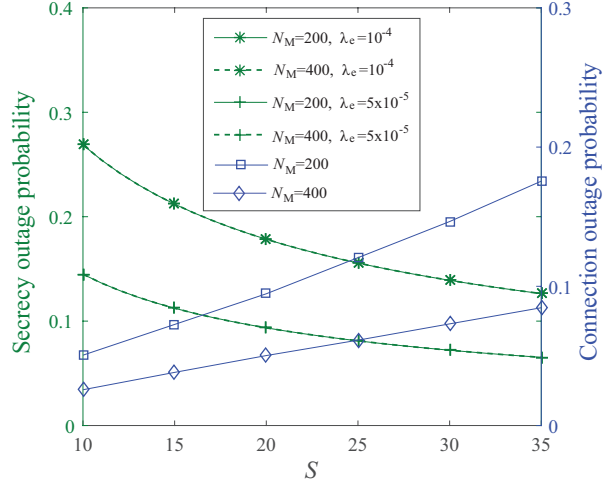}
        \caption{The secrecy outage probability and connection outage probability for MBS transmission in delay-limited mode: $R= 8$ bits/s/Hz, $R_\mathrm{s}=0.3 R$, $d_o=50$ m, $\lambda_\mathrm{M}=\left(500^2 \times \pi\right)^{-1} \mathrm{m}^{-2}$, $\lambda_\mathrm{R}= 20 \times \lambda_\mathrm{M}$,  $\eta_\mathrm{M}=3.0$ and $\eta_\mathrm{R} = 3.6$.}
        \label{Fig2}
    \end{center}
\end{figure}
Fig.~\ref{Fig2} provides the results for the secrecy outage probability and connection outage probability of MBS transmission operating in delay-limited mode. With increasing the number of served users, the secrecy outage probability decreases and  the connection outage probability increases. The reason is that the transmit power allocated to each user data stream decreases when serving more users at the MBS, which in turn decreases the receive SINR at both the legitimate user and eavesdroppers. To decrease the connection outage probability without altering the secrecy outage probability, MBSs can be equipped with more antennas to provide larger array gains for the legitimate users. In addition, it is obvious that more eavesdroppers will deteriorate the secrecy performance.

\begin{figure}
    \begin{center}
        \includegraphics[width=3.3 in]{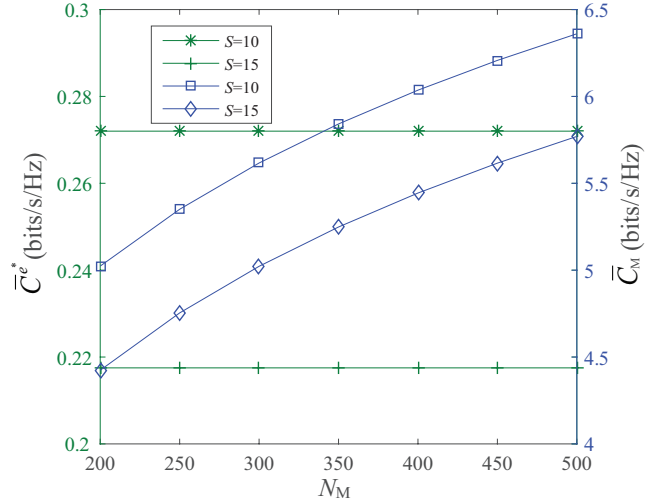}
        \caption{The ergodic capacity $\bar{C}^{e^*}$ of the most malicious eavesdropper's channel  and the ergodic capacity $\bar{C}_\mathrm{M}$ of the macrocell user's channel  for MBS transmission in delay-tolerant mode: $\lambda_\mathrm{M}=\left(500^2 \times \pi\right)^{-1} \mathrm{m}^{-2}$, $\lambda_\mathrm{R}= 20 \times \lambda_\mathrm{M}$, $\lambda_e=10^{-5}$ m$^{-2}$, $\eta_\mathrm{M}=3.3$.}
        \label{Fig3}
    \end{center}
\end{figure}

Fig.~\ref{Fig3} shows the ergodic capacity $\bar{C}^{e^*}$ of the most malicious eavesdropper's channel  and the ergodic capacity $\bar{C}_\mathrm{M}$
  of the macrocell user's
channel for MBS transmission in delay-tolerant mode. %For MBS transmission, $\bar{C}^{{e^{\rm{*}}}}=\bar{C}_{\mathrm{R},i}^{e^*} \left(i \in \left\{k,\nu\right\}\right)$ (from \eqref{Eve_Ergodic_capacity_overline}) and $\bar{C}^{{e^{\rm{*}}}}=\bar{C}_{\mathrm{M},\nu}^{e^*}$ (from \eqref{Eve_MBS_rate}).
When adding more MBS antennas, $\bar{C}^{e^*}$ is unaltered and $\bar{C}_\mathrm{M}$ experiences a substantial increase,  since only the legitimate macrocell users can obtain the array gains. Additionally, serving more users at the MBS decreases $\bar{C}^{{e^{\rm{*}}}}$ and $\bar{C}_\mathrm{M}$, because of lower transmit power per user data stream at the MBS as mentioned in  Fig~\ref{Fig2}.
\begin{figure}
    \begin{center}
        \includegraphics[width=3.1 in]{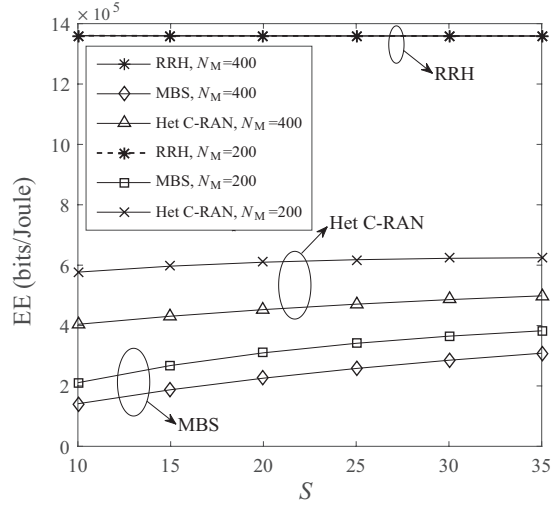}
        \caption{Effects of massive MIMO on the EE: $\lambda_\mathrm{M}=\left(500^2 \times \pi\right)^{-1} \mathrm{m}^{-2}$, $\lambda_\mathrm{R}= 20 \times \lambda_\mathrm{M}$, $\eta_\mathrm{M}=3.2$, $\eta_\mathrm{R} = 3.6$, and $\alpha=0.5$.}
        \label{Fig4}
    \end{center}
\end{figure}

Fig.~\ref{Fig4} illustrates the effects of massive MIMO on the EE. Results indicate that the EE of MBS transmission is improved by serving more users at the MBS, which is attributed to the fact that more multiplexing gain is achieved. Although adding more antennas at the MBS can provide a large array gain, there is a significant increase in power consumption resulting from massive MIMO baseband processing, which decreases the EE of MBS transmission. In addition, the RRHs achieve  higher EE than the MBSs, as they use lower transmit power and do not consume power for baseband processing, and results also demonstrate that massive MIMO has negligible effect on the EE of RRH transmission.

\subsection{The Effects of RRH Density}

\begin{figure}
    \begin{center}
        \includegraphics[width=3.3 in]{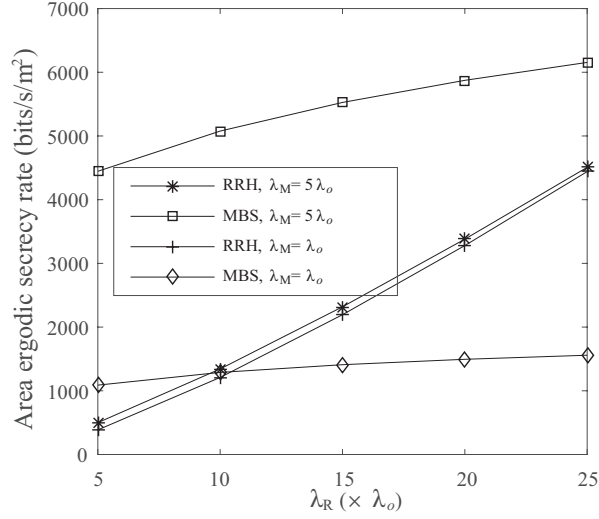}
        \caption{Effects of RRH density on area ergodic secrecy rate: $\lambda_o=\left(500^2 \times \pi\right)^{-1} \mathrm{m}^{-2}$,  $\lambda_e=10^{-5} \mathrm{m}^{-2}$, $N_\mathrm{M}=400$, $S=30$, $\eta_\mathrm{M}=3.0$, $\eta_\mathrm{R} = 3.6$, and $\alpha=0.7$.}
        \label{Fig5}
    \end{center}
\end{figure}
\begin{figure}[t!]
\centering
\subfigure[ $R_s=0.2 R$, $d_o=30$ m, $\lambda_o=\left(500^2 \times \pi\right)^{-1} \mathrm{m}^{-2}$, $\lambda_e=10^{-4} \mathrm{m}^{-2}$, $\eta_\mathrm{R} = 3.6$.]{
\label{Fig.sub.61}
\includegraphics[width=2.8 in]{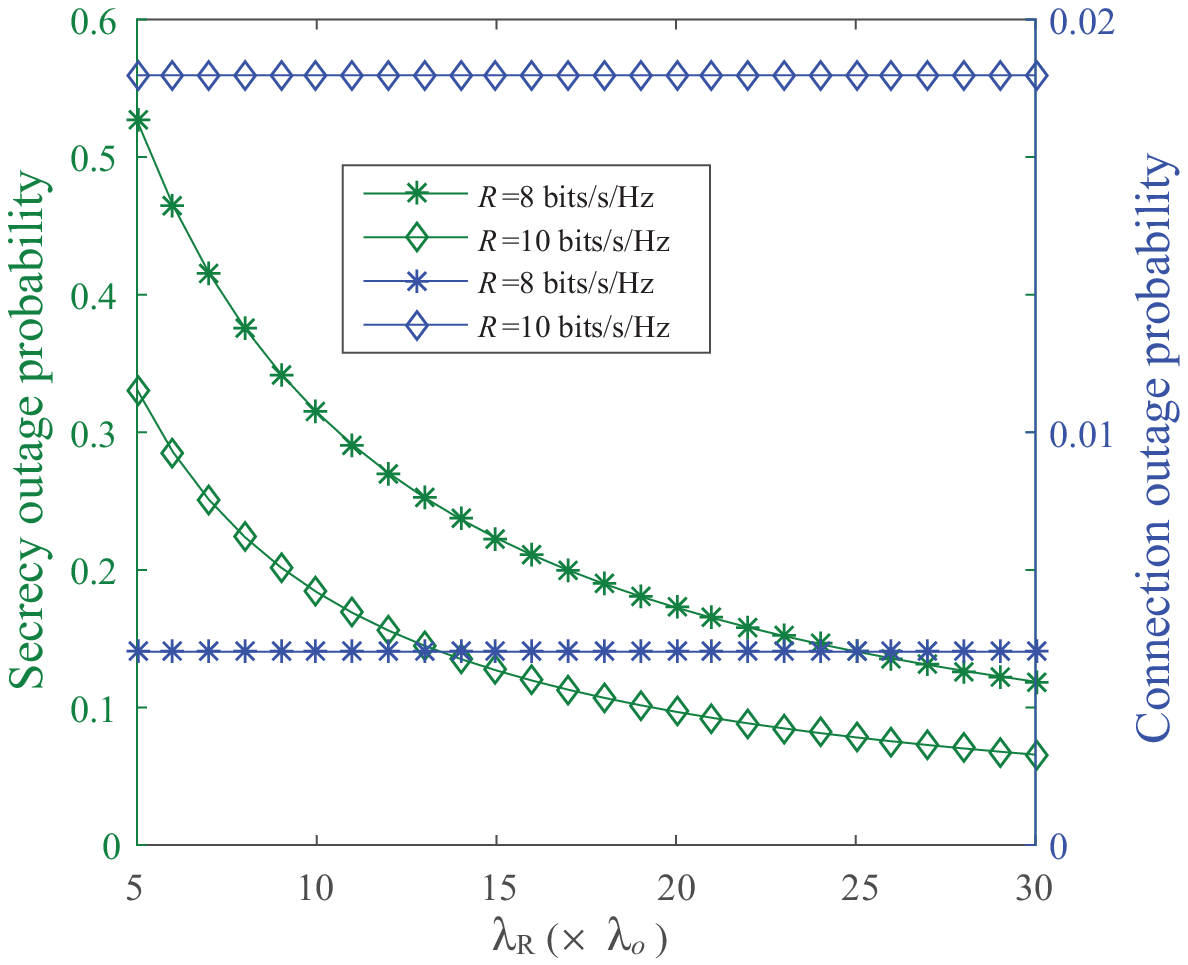}}
\subfigure[ $R_s=0.2 R$, $d_o=30$ m, $\lambda_\mathrm{M}=\lambda_o=\left(500^2 \times \pi\right)^{-1} \mathrm{m}^{-2}$, $\lambda_e=10^{-4} \mathrm{m}^{-2}$, $N_\mathrm{M}=200$, $S=15$, $\eta_\mathrm{M}=3.0$, and $\eta_\mathrm{R} = 3.6$.]{
\label{Fig.sub.62}
\includegraphics[width=2.8 in]{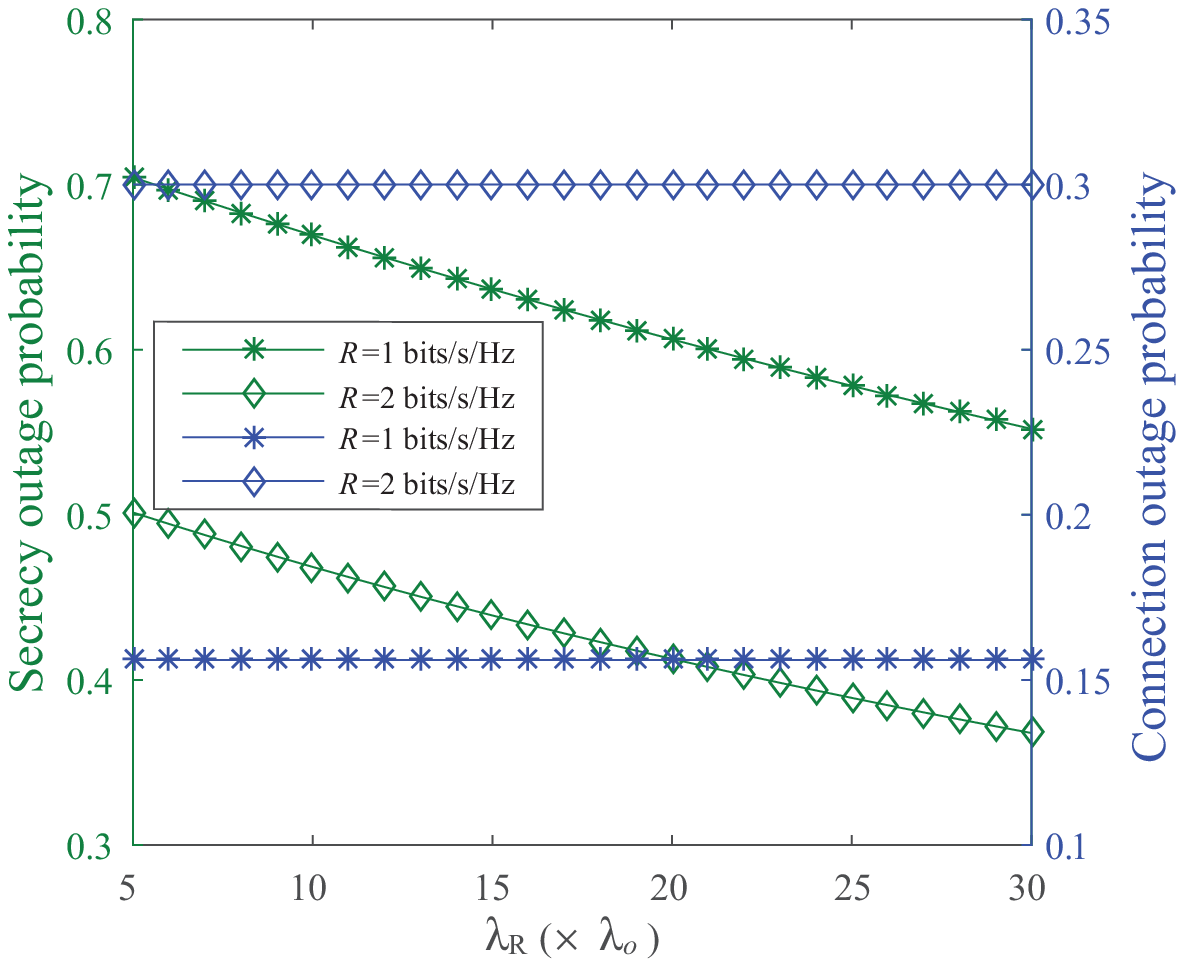}}
\caption{Secrecy outage probability and connection outage probability for RRH transmission in delay-limited mode.}
\label{Fig6}
\end{figure}

Fig.~\ref{Fig5} shows the effects of RRH density on area ergodic secrecy rate. We observe that when more RRHs are deployed, there is a substantial increase in the area ergodic secrecy rate of the RRH tier, as illustrated in \textbf{Remark 2} of Section III-A.  The area ergodic secrecy rate of the MBS tier can also increase with the density of RRH, since users with far-away MBSs will be offloaded to the RRHs. RRH tier can achieve higher area ergodic secrecy rate than the massive MIMO aided MBS tier when the RRHs are denser than the MBSs. In addition, slightly increasing the number of massive MIMO macrocells brings large improvement in the area ergodic secrecy rate of the MBS tier because  more users can be served, and it also improves the area ergodic secrecy rate of the RRH tier due to the fact that users with far-away RRHs will be offloaded to the MBSs.

Fig.~\ref{Fig6} shows the secrecy outage probability and connection outage probability of RRH transmission in delay-limited mode.   Specifically, Fig.~\ref{Fig.sub.61} focuses on the performance when RRH transmissions operate over the RBs only allocated to RRHs, while Fig.~\ref{Fig.sub.62} concentrates on the performance when RRH transmissions operate over the RBs shared by RRHs and MBSs. As stated in   \textbf{Remark 5} of Section III-B, the secrecy outage probability experiences a massive decline when increasing the density of RRHs, due to more severe interference on the Eves but the connection outage probability is unaltered since the inter-RRH interference is mitigated in the C-RAN, as mentioned in \textbf{Remark 4}. Compared with the use of RBs shared by RRHs and MBSs, RRH achieves better performance by using the RBs only used by RRHs, due to the absence of inter-tier interference in these RBs.

Fig.~\ref{Fig7} shows the ergodic capacity $\bar{C}^{e^*}$ of the most malicious eavesdropper's channel  and the ergodic capacity $\bar{C}_\mathrm{R}$ of the C-RAN user's channel  for RRH transmission in delay-tolerant mode. As suggested in \textbf{Remark 2}, deploying more RRHs can significantly decrease $\bar{C}^{e^*}$ and increase $\bar{C}_\mathrm{R}$. For RRH transmission over the RB shared by RRHs and MBSs, interference from the MBS tier has a large negative impact on the performance at the legitimate users, however, its impact on the degradation of the most malicious eavesdropper's channel is limited compared to more interference from dense RRHs.

\begin{figure}
    \begin{center}
        \includegraphics[width=3.3 in]{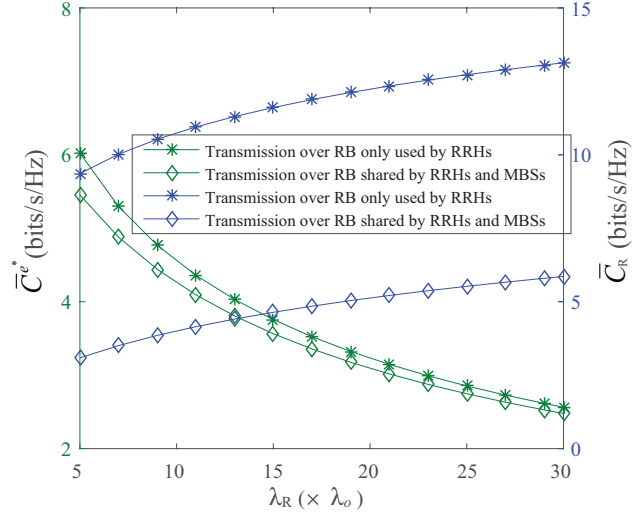}
        \caption{The ergodic capacity $\bar{C}^{e^*}$ of the most malicious eavesdropper's channel  and the ergodic capacity $\bar{C}_\mathrm{R}$ of the C-RAN user's channel  for RRH transmission in delay-tolerant mode: $\lambda_\mathrm{M}=\lambda_o=\left(500^2 \times \pi\right)^{-1} \mathrm{m}^{-2}$, $\lambda_e=10^{-4}$ m$^{-2}$, $N_\mathrm{M}=200$, $S=15$, $\eta_\mathrm{M}=3.5$, and $\eta_\mathrm{R} = 3.2$.}
        \label{Fig7}
    \end{center}
\end{figure}

\begin{figure}
    \begin{center}
        \includegraphics[width=3.3 in]{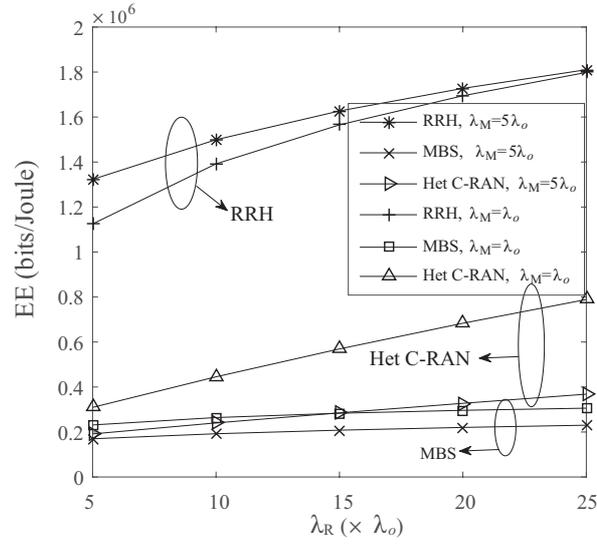}
        \caption{Effects of RRH density on the EE: $\lambda_o=\left(500^2 \times \pi\right)^{-1} \mathrm{m}^{-2}$, $N_\mathrm{M}=400$, $S=30$, $\eta_\mathrm{M}=3.0$, $\eta_\mathrm{R} = 3.6$, and $\alpha=0.7$.}
        \label{Fig8}
    \end{center}
\end{figure}

Fig.~\ref{Fig8} shows the effects of RRH density on the EE. As mentioned in {\textbf{Corollary 2}}  of  Section IV, when increasing the density of RRHs, the EE of RRH transmission is significantly improved. Increasing the density of MBSs improves the EE of RRH transmission but decreases the EE of MBS transmission, since users with far-away RRHs are offloaded to the MBSs. The EE of the Het C-RAN decreases with increasing the density of MBSs, due to the fact that power consumption of the network is significantly boosted by using more massive MIMO MBSs. Since RRHs achieve higher EE, more RRHs should be deployed in the Het C-RAN to enhance the EE.

\subsection{The Effects of S-FFR}
\begin{figure}
    \begin{center}
        \includegraphics[width=4.0 in]{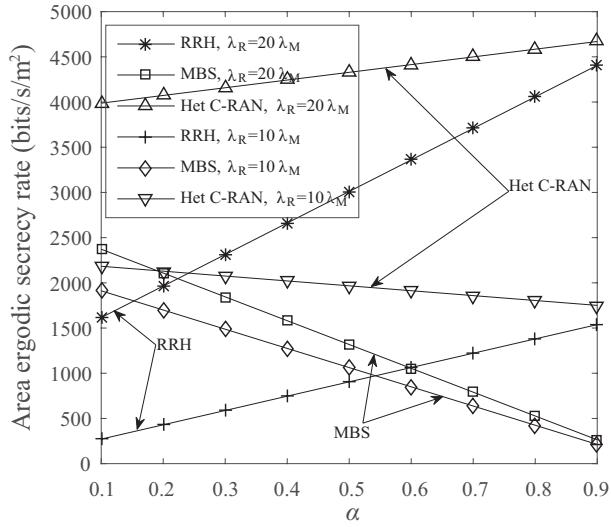}
        \caption{Effects of S-FFR on area ergodic secrecy rate: $\lambda_\mathrm{M}=\left(500^2 \times \pi\right)^{-1} \mathrm{m}^{-2}$, $\lambda_e=5*10^{-5}$ m$^{-2}$, $N_\mathrm{M}=400$, $S=25$, $\eta_\mathrm{M}=3.5$, and $\eta_\mathrm{R} = 3.3$.}
        \label{Fig9}
    \end{center}
\end{figure}
\begin{figure}
    \begin{center}
        \includegraphics[width=3.3 in]{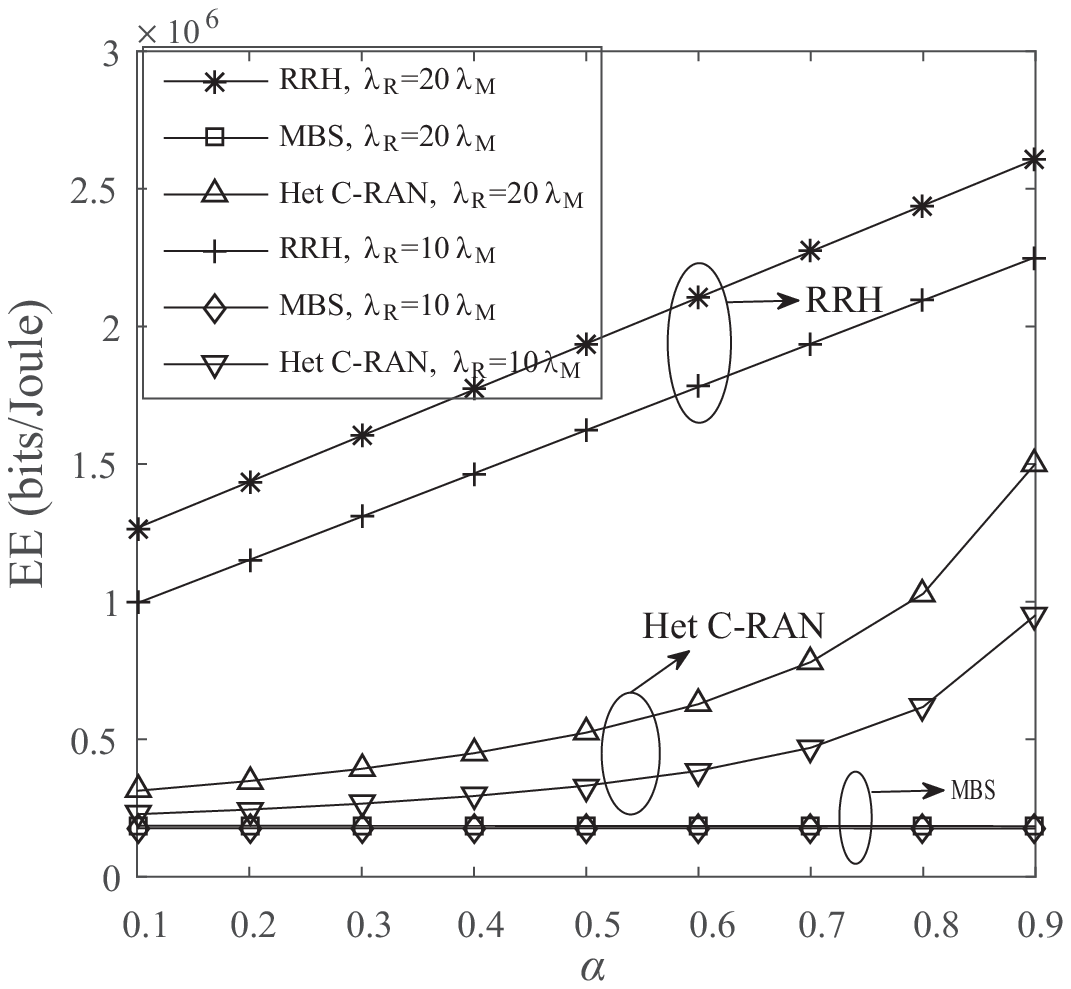}
        \caption{Effects of S-FFR on the EE: $\lambda_\mathrm{M}=\left(500^2 \times \pi\right)^{-1} \mathrm{m}^{-2}$, $\lambda_e=5*10^{-5}$ m$^{-2}$, $N_\mathrm{M}=400$, $S=25$, $\eta_\mathrm{M}=3.5$, and $\eta_\mathrm{R} = 3.3$.}
        \label{Fig10}
    \end{center}
\end{figure}

Results in Fig.~\ref{Fig9} demonstrate the effects of S-FFR on area ergodic secrecy rate. It is obvious that with more RBs allocated to the RRHs, the area ergodic secrecy rate increases for the RRH tier, and decreases for the MBS tier. The RRH tier can achieve higher area ergodic secrecy rate than the MBS tier, when the density of RRHs and the allocated RBs are large. More importantly, it is implied that the effect of S-FFR on the area ergodic secrecy rate of the network can be distinct depending on the RRH density.

Finally, Fig.~\ref{Fig10} provides the effects of S-FFR on the EE. As mentioned in section IV, the EE for RRH transmission is indeed linearly improved by allocating more RBs to the RRHs without the harm of inter-tier interference. S-FFR indeed has little effect on the EE of MBS transmission. Therefore, the EE of the network increases with the RRH density  and RBs only used by RRHs, as shown in this figure.

\section{Conclusions}
In this paper, we investigated the physical layer secrecy and EE in the two-tier massive MIMO aided heterogeneous C-RAN, where massive MIMO empowered macrocell BSs and RRHs coexist. The implementation of S-FFR was utilized to suppress the inter-tier interference. We first studied the impacts of massive MIMO and C-RAN on the secrecy performance in terms of the area ergodic secrecy rate and secrecy outage probability. Then we evaluated the EE in such networks. Our results demonstrated that both C-RAN and massive MIMO can significantly enhance the secrecy performance. The implementation of C-RAN with low power cost RRHs improves EE of the networks substantially.

\section*{Appendix A: A detailed derivation of Theorem 1}
\label{App:theo_1}
\renewcommand{\theequation}{A.\arabic{equation}}
\setcounter{equation}{0}
When using the $k$-th RB allocated to the RRHs, the ergodic capacity of the channel between the typical RRH and its served user is given by
\begin{align}\label{A1_1}
&\bar{C}_{\mathrm{R},k}={\mathbb{E}}\left\{\log_2\left(1+\gamma_{\mathrm{R},k}\right)\right\}\nonumber\\
&=\int_0^\infty {\mathbb{\mathrm{E}}}_{h_{\mathrm{R},k}} \left\{\log_2\left(1+\frac{P_\mathrm{R}\beta}{B_oN_o}h_{\mathrm{R},k} x^{-\eta_\mathrm{R}}\right)\right\} f_{\left| {{X_{{o},\mathrm{R}}}}\right|} \left(x\right) dx.
\end{align}
Considering that $h_{\mathrm{R},k} \sim \rm{exp}(1)$, we further have
\begin{align}\label{A1_2}
\bar{C}_{\mathrm{R},k}&=\frac{1}{{\ln 2}}\int_0^\infty  {\left\{ {\int_{\rm{0}}^\infty  {\frac{1}{{1 + t}}} {e^{ - \frac{{{B_o}{N_o}}}{{{P_{\rm{R}}}\beta }}{x^{{\eta _{\rm{R}}}}}t}}dt} \right\}{f_{\left| {{X_{o,{\rm{R}}}}} \right|}}\left( x \right)dx} \nonumber\\
&= \frac{1}{{\ln 2}}\int_0^\infty  {{e^{\frac{{{B_o}{N_o}}}{{{P_{\rm{R}}}\beta }}{x^{{\eta _{\rm{R}}}}}}}\Gamma \left( {0,\frac{{{B_o}{N_o}}}{{{P_{\rm{R}}}\beta }}{x^{{\eta _{\rm{R}}}}}} \right){f_{\left| {{X_{o,{\rm{R}}}}} \right|}}\left( x \right)dx},
\end{align}
where ${f_{\left| {{X_{o,{\rm{R}}}}} \right|}}\left( x \right)$ is the probability density function (PDF) of the distance between the typical RRH and its intended user, using the similar approach in \cite{HS12},  ${f_{\left| {{X_{o,{\rm{R}}}}} \right|}}\left( x \right)$ is given by
\begin{align}\label{A1_3}
{f_{\left| {{X_{o,{\rm{R}}}}} \right|}}\left( x \right)=\frac{{2\pi {\lambda_\mathrm{R}}}}{{{\mathcal{A}_\mathrm{R}}}}x{e^{ - \pi \left( {{\lambda _\mathrm{R}} + {\lambda _\mathrm{M}}} \right){x^2}}},
\end{align}
where ${\mathcal{A}_\mathrm{R}}=\frac{\lambda _\mathrm{R}}{\lambda _\mathrm{R}+\lambda _\mathrm{M}}$ is the probability that
 a user is associated with the RRH. By plugging \eqref{A1_3} into \eqref{A1_2}, we get \eqref{RRH_Ergodic_rate_1}.

 When using the $\nu$-th RB shared by the RRHs and MBSs, the ergodic capacity of the channel between the typical RRH and its served user is given by
\begin{align}\label{A1_4}
\bar{C}_{\mathrm{R},\nu}&={\mathbb{E}}\left\{\log_2\left(1+\gamma_{\mathrm{R},\nu}\right)\right\}\nonumber\\
&=\frac{1}{{\ln 2}}\int_0^\infty  {{\mathbb{E}}_{\left| {{X_{o,{\rm{R}}}}} \right|=x}\left\{\log_2\left(1+\gamma_{\mathrm{R},\nu}\right)\right\}{f_{\left| {{X_{o,{\rm{R}}}}} \right|}}\left( x \right)dx}\nonumber\\
&=\frac{1}{{\ln 2}}\int_0^\infty  {\left[ {\int_0^\infty  {\frac{{\bar{F}}_{{\gamma_{\mathrm{R},\nu}}\mid \left\{\left| {{X_{o,{\rm{R}}}}} \right|=x\right\}}\left(\gamma\right)}{{1 + \gamma }}d\gamma } } \right]{f_{\left| {{X_{o,{\rm{R}}}}} \right|}}\left( x \right)dx},
\end{align}
where ${\bar{F}}_{{\gamma_{\mathrm{R},\nu}}\mid \left\{\left| {{X_{o,{\rm{R}}}}} \right|=x\right\}}\left(\gamma\right)$ is the CCDF of $\gamma_{\mathrm{R},\nu}$ given a distance $\left| {{X_{o,{\rm{R}}}}} \right|=x$, which is calculated as
\begin{align}\label{A1_5}
&{\bar{F}}_{{\gamma_{\mathrm{R},\nu}}\mid \left\{\left| {{X_{o,{\rm{R}}}}} \right|=x\right\}}\left(\gamma\right)=\Pr \left( {\frac{{{P_{\rm{R}}}{h_{{\rm{R}},\nu }}\beta {x^{ - {\eta _{\rm{R}}}}}}}{{{I_{{\rm{M}},\nu }} + {B_o}{N_o}}} > \gamma } \right)\nonumber\\
&= {e^{ - \frac{{{B_o}{N_o}}}{{{P_{\rm{R}}}\beta }}{x^{{\eta _{\rm{R}}}}}\gamma }}{\mathbb{\mathrm{E}}_{{\Phi _\mathrm{M}}}}\left\{ {{e^{ - \frac{1}{{{P_{\rm{R}}}\beta }}{x^{{\eta _{\rm{R}}}}}\gamma {I_{{\rm{M}},\nu }}}}} \right\}\nonumber\\
&={e^{ - \frac{{{B_o}{N_o}}}{{{P_{\rm{R}}}\beta }}{x^{{\eta _{\rm{R}}}}}\gamma }} \mathcal{L}_{{I_{{\rm{M}},\nu }}}\left({\frac{1}{{{P_{\rm{R}}}\beta }}{x^{{\eta _{\rm{R}}}}}\gamma }\right),
\end{align}
where $\mathcal{L}_{{I_{{\rm{M}},\nu }}}\left(\cdot\right)$ is the laplace transform of the PDF of ${I_{{\rm{M}},\nu }}$, and is  given by
\begin{align}\label{A1_6}
&\mathcal{L}_{{I_{{\rm{M}},\nu }}}\left(s\right)=\mathbb{\mathrm{E}}\left\{\exp\left\{-\left(\sum\limits_{\ell  \in {\Phi_\mathrm{M}}} {\frac{P_\mathrm{M}}{S}{h_{\ell,\nu} }\beta \left| {{X_{\ell ,\mathrm{M}}}} \right|^{-\eta_\mathrm{M}}}\right)s\right\}\right\}\nonumber\\
&\mathop = \limits^{\left(b\right)} \exp\left\{{ - \int_x^\infty  {\left( {1 - \frac{1}{{{{\left( {1 + s\frac{{{P_{\mathrm{M}}}}}
{S}\beta {r^{ - {\eta_\mathrm{M}}}}} \right)}^S}}}} \right){\lambda_\mathrm{M}}2\pi rdr} }\right\}\nonumber\\
&\mathop = \limits^{\left(c\right)} \exp \left( { - {\lambda _\mathrm{M}}2\pi \sum\limits_{\mu  = 1}^S {
S \choose \mu}\int_x^\infty  {\frac{{{{ {{\left( {\frac{{{P_{\rm{M}}}}}{S}\beta } \right)}^\mu }{s^\mu }
\left( {{r^{ - {\eta_\mathrm{M}}}}}
\right)}^\mu }}}{{{{\left( {1 + s\frac{{{P_{\mathrm{M}}}}}{S}\beta {r^{ - {\eta_\mathrm{M}}}}} \right)}^S}}}rdr} } \right)\nonumber\\
&=\exp \bigg\{   - {\lambda _\mathrm{M}}2\pi \sum\limits_{\mu  = 1}^S
{{
S\choose
\mu
}{{\left( {s  \frac{{ {P_{\mathrm{M}}}}}{{S}}}\beta \right)}^\mu }}  \frac{{{{\left( { - s \frac{{
{P_{\mathrm{M}}}}}{{S}}}  \beta \right)}^{ - \mu  + \frac{2}{{{\eta _\mathrm{M}}}}}}}}{{{\eta_\mathrm{M}}}}\nonumber\\
&\;\;\; {B_{\left(- {s \frac{{ {P_{\mathrm{M}}}}}{{S}} \beta
x^{-\eta _\mathrm{M}}} \right)}}\left[
{\mu  - \frac{2}{{{\eta_\mathrm{M}}}},1 - S} \right]\bigg\},
\end{align}
where (b) is obtained by using the generating functional of PPP~\cite{M_Haenggi2013},  (c) results from using Binomial expansion, $\mathrm{B}_{\left(\cdot\right)}\left[\cdot,\cdot\right]$ is the incomplete beta function~\cite[(8.391)]{gradshteyn}. By pulling \eqref{A1_6} and \eqref{A1_5} together, we get \eqref{CDF_R_secrecy_rate}. Substituting \eqref{A1_3} into \eqref{A1_4}, we also get \eqref{RRH_Ergodic_rate_2}.

\section*{Appendix B: A detailed derivation of Theorem 2}
\label{App:theo_1}
\renewcommand{\theequation}{B.\arabic{equation}}
\setcounter{equation}{0}

The ergodic capacity $\bar{C}_{\mathrm{R},i}^{e^*}$ $(i\in\left\{k,\nu\right\})$ of the most malicious eavesdropper's channel is written as
\begin{align}\label{B1_1}
\bar{C}_{\mathrm{R},i}^{e^*}&={\mathbb{E}}\left\{\log_2\left(1+\gamma_{\mathrm{R},i}^{e^*}\right) \right\} \nonumber\\
&=\frac{1}{{\ln 2}}\int_0^\infty \frac{1-F_{\gamma_{\mathrm{R},i}^{e^*}}\left(x\right)}{1+x} dx,
\end{align}
where $F_{\gamma_{\mathrm{R},i}^{e^*}}\left(x\right)$ denotes the cumulative distribution function (CDF) of ${\gamma_{\mathrm{R},i}^{e^*}}$.

Based on \eqref{gamma_e_RRH_k_SINR}, the CDF of ${\gamma_{\mathrm{R},k}^{e^*}}$ is calculated as
 \begin{align}\label{CDF_e_RRH_k_SINR}
 & F_{\gamma_{\mathrm{R},k}^{e^*}}\left(x\right)=\Pr\left( \gamma_{\mathrm{R},k}^{e^*}< x \right) \nonumber\\
 &=\Pr\left(\mathop {\max }\limits_{e \in \Phi_e}\left\{\frac{{{P_\mathrm{R}}h_{\mathrm{R}{\rm{,}}k}^e\beta {{\left| {X_{o{\rm{,}}\mathrm{R}}^e} \right|}^{ - {\eta _\mathrm{R}}}}}}{{I_{\mathrm{R},k}^e + {B_o}{N_e}}}\right\}< x \right)\nonumber\\
&=\mathbb{E}_{\Phi_e} \bigg\{ \prod\limits_{e \in {\Phi_e}} {\Pr\left(\frac{{{P_\mathrm{R}}h_{\mathrm{R}{\rm{,}}k}^e\beta {{\left| {X_{o{\rm{,}}\mathrm{R}}^e} \right|}^{ - {\eta _\mathrm{R}}}}}}{{ I_{\mathrm{R},k}^e + {B_o}{N_e}}}<x  \left.\right|  \Phi_e     \right) } \bigg\}.
 \end{align}
Using the generating functional of the PPP $\Phi_e$, $F_{\gamma_{\mathrm{R},k}^{e^*}}\left(x\right)$ can be further derived as
\begin{align}\label{CDF_Eve_Gamma_App_B}
& F_{\gamma_{\mathrm{R},k}^{e^*}}\left(x\right) =\exp\bigg\{-{\lambda _e}\int_{{R^{\rm{2}}}}\Big(1-{\Pr\Big(\frac{{{P_\mathrm{R}}h_{\mathrm{R}{\rm{,}}k}^e\beta {{r}^{ - {\eta _\mathrm{R}}}}}}{{ I_{\mathrm{R},k}^e + {B_o}{N_e}}}<x  \Big) }\Big)   {dr} \Big\}\nonumber\\
&=\exp\bigg\{-{\lambda _e}\int_{{R^{\rm{2}}}}\mathbb{E}_{\Phi_{\mathrm{R}}} \bigg\{ \mathbb{E}_{\Phi_\mathrm{M}} \bigg\{ \nonumber\\
&\quad\quad\quad\quad\quad\;\;\exp\big[-\frac{{r}^{ {\eta _\mathrm{R}}}x}{{P_\mathrm{R}}\beta}\left( I_{\mathrm{R},k}^e + {B_o}{N_e}\right)\big]\bigg\}\bigg\}dr\bigg\}\nonumber\\
&\mathop = \limits^{\left(a\right)}\exp\bigg\{-2\pi{\lambda _e}\int_0^\infty  \exp\big[-\frac{{r}^{ {\eta _\mathrm{R}}}x}{{P_\mathrm{R}}\beta}{B_o}{N_e}\big] \mathcal{L}_{I_{\mathrm{R},k}^e}\big(\frac{{r}^{ {\eta _\mathrm{R}}}x}{{P_\mathrm{R}}\beta}\big) r dr\bigg\},
\end{align}
where (a) results from using the polar-coordinate system, $\mathcal{L}_{I_{\mathrm{R},k}^e}\left(\cdot\right)$ is the laplace transform of the PDF of ${I_{\mathrm{R},k}^e}$.

Likewise, the CDF of ${\gamma_{\mathrm{R},\nu}^{e^*}}$ is calculated as
\begin{align}\label{CDF_Eve_Gamma}
& F_{\gamma_{\mathrm{R},\nu}^{e^*}}\left(x\right)=\Pr\left( \gamma_{\mathrm{R},\nu}^{e^*}< x \right) \nonumber\\
&=\Pr\left(\mathop {\max }\limits_{e \in \Phi_e}\left\{\frac{{{P_\mathrm{R}}h_{\mathrm{R}{\rm{,}}\nu}^e\beta {{\left| {X_{o{\rm{,}}\mathrm{R}}^e} \right|}^{ - {\eta _\mathrm{R}}}}}}{{I_{\mathrm{R},\nu}^e +I_{\mathrm{M},\nu}^e + {B_o}{N_e}}}\right\}< x \right)\nonumber\\
&= \mathbb{E}_{\Phi_e} \bigg\{ \prod\limits_{e \in {\Phi_e}} {\Pr\left(\frac{{{P_\mathrm{R}}h_{\mathrm{R}{\rm{,}}\nu}^e\beta {{\left| {X_{o{\rm{,}}\mathrm{R}}^e} \right|}^{ - {\eta _\mathrm{R}}}}}}{{I_{\mathrm{M},\nu}^e + I_{\mathrm{R},\nu}^e + {B_o}{N_e}}}<x  \left.\right|  \Phi_e     \right) } \bigg\} \nonumber\\
&=\exp\bigg\{-2\pi{\lambda _e}\int_0^\infty  \exp\big[-\frac{{r}^{ {\eta _\mathrm{R}}}x}{{P_\mathrm{R}}\beta}{B_o}{N_e}\big] \nonumber\\
&\quad\quad\quad\quad\quad\quad\quad\quad\mathcal{L}_{I_{\mathrm{R},\nu}^e}\big(\frac{{r}^{ {\eta _\mathrm{R}}}x}{{P_\mathrm{R}}\beta}\big) \mathcal{L}_{I_{\mathrm{M},\nu}^e}\big(\frac{{r}^{ {\eta _\mathrm{R}}}x}{{P_\mathrm{R}}\beta}\big) r dr\bigg\},
\end{align}
where  $\mathcal{L}_{I_{\mathrm{R},\nu}^e}\left(\cdot\right)$ and $\mathcal{L}_{I_{\mathrm{M},\nu}^e}\left(\cdot\right)$ are the laplace transforms of the PDFs of ${I_{\mathrm{R},\nu}^e}$ and ${I_{\mathrm{M},\nu}^e}$, respectively.

By using the Slivnyak's theorem and the generating functional of the PPP $\Phi_R$, $\mathcal{L}_{I_{\mathrm{R},i}^e}\left(\cdot\right)$ $(i\in\left\{k,\nu\right\})$ is given by
\begin{align}\label{L_I_R_RRH}
&\mathcal{L}_{I_{\mathrm{R},i}^e}\left(s\right)=\mathbb{E}\left\{\exp\left(-s I_{\mathrm{R},i}^e\right) \right\} \nonumber\\
&=\exp\left(-2\pi \lambda_\mathrm{R}\int_0^\infty \left(1-\frac{1}{\left(1+s P_\mathrm{R} \beta r^{-\eta_\mathrm{R}} \right)}\right)r dr    \right)\nonumber\\
&=\exp\left(-\lambda_\mathrm{R}\pi \left(P_\mathrm{R}\beta\right)^{\frac{2}{\eta _\mathrm{R}}}\Gamma\left(1+\frac{2}{\eta _\mathrm{R}}\right) \Gamma\left(1-\frac{2}{\eta _\mathrm{R}}\right) s^{\frac{2}{\eta _\mathrm{R}}} \right).
\end{align}
Similarly, ${I_{\mathrm{M},\nu}^e}$ is given by
\begin{align}\label{Lap_I_M_e}
\hspace{-0.5 cm}\mathcal{L}_{I_{\mathrm{M},\nu}^e}\left(s\right)&=\exp\left[-2\pi \lambda_\mathrm{M}\int_0^\infty \left(1-\frac{1}{\left(1+s\frac{P_\mathrm{M}}{S} \beta r^{-\eta_\mathrm{M}} \right)^S}\right)r dr    \right] \nonumber\\
&=\exp\bigg[-2\pi \lambda_\mathrm{M} \sum\limits_{\mu  = 1}^S
{S\choose \mu} \left(s\frac{P_\mathrm{M}}{S}\beta\right)^{\frac{2}{\eta_\mathrm{M}}}\nonumber\\
&\quad\quad\quad \quad \frac{\Gamma\left(\mu -\frac{2}{\eta_\mathrm{M}}\right) \Gamma\left(-\mu +\frac{2}{\eta_\mathrm{M}}+S\right)}{\eta_\mathrm{M}\Gamma\left(S\right)}\bigg].
\end{align}
Substituting \eqref{L_I_R_RRH} into \eqref{CDF_Eve_Gamma_App_B}, we get $F_{\gamma_{\mathrm{R},k}^{e^*}}\left(\cdot\right)$ as \eqref{X3_34_overline_F_k}. Then, substituting \eqref{L_I_R_RRH} and \eqref{Lap_I_M_e} into \eqref{CDF_Eve_Gamma}, we get $F_{\gamma_{\mathrm{R},\nu}^{e^*}}\left(\cdot\right)$ as \eqref{X3_34_overline_F}.

\section*{Appendix C: A detailed derivation of Theorem 3}
\label{App:theo_1}
\renewcommand{\theequation}{C.\arabic{equation}}
\setcounter{equation}{0}
The ergodic capacity of the channel between the typical MBS and its served user is written as
\begin{align}\label{C1_1}
\bar{C}_{\mathrm{M},\nu}=\mathbb{E}\left\{\log_2\left(1+\gamma_{\mathrm{M},\nu}\right)\right\}.
\end{align}
By using Jensen's inequality,  a tight lower bound for $\bar{C}_{\mathrm{M},\nu}$ is given by~\cite{Lifeng_massiveMIMO}
\begin{align}\label{C1_2}
\bar{C}_{\mathrm{M},\nu}^{\mathrm{L}}=\log_2\left(1+e^{Z_3
+Z_4}\right),
\end{align}
where
\begin{align}\label{C1_3}
Z_3=\mathbb{\mathrm{E}}\left\{\ln\left(\frac{P_\mathrm{M}}{S}g_{\mathrm{M},\nu}\beta \left| {{X_{{o},\mathrm{M}}}} \right|
^{-\eta_\mathrm{M}}\right)\right\},
\end{align}
and
\begin{align}\label{C1_4}
Z_4=\mathbb{\mathrm{E}}\left\{\ln\left(\frac{1}{J_{\mathrm{M},\nu}+J_{\mathrm{R},\nu}+B_o N_1}\right)\right\}.
\end{align}

We first calculate $Z_3$ as
\begin{align}\label{Z_1_theo2}
Z_3=\ln \left( {\frac{{{P_{\rm{M}}}}}{S}\beta } \right){\rm{ + }}\mathbb{\mathrm{E}}\left\{ {\ln \left( {{g_{{\rm{M}},\nu }}} \right)} \right\} - {\eta _{\rm{M}}}\mathbb{\mathrm{E}}\left\{ {\ln \left( {\left| {{X_{o,{\rm{M}}}}} \right|} \right)} \right\}.
\end{align}
Considering that $g_{\mathrm{M},\nu} \sim \Gamma\left(N_\mathrm{M}-S+1,1\right)$, $\mathbb{\mathrm{E}}\left\{ {\ln \left( {{g_{{\rm{M}},\nu }}} \right)} \right\}$ is given by
\begin{align}\label{Exp_ln_Z1}
\mathbb{\mathrm{E}}\left\{ {\ln \left( {{g_{{\rm{M}},\nu }}} \right)} \right\} &= \int_0^\infty  {\frac{{{x^{N_\mathrm{M} - S}}{e^{ - x}}}}{{\left( {N_\mathrm{M} - S} \right)!}}\ln \left( x \right)dx} \nonumber\\
& \mathop = \limits^{\left(a\right)} \psi\left(N_\mathrm{M}-S+1\right),
\end{align}
where (a) results from using $\int_0^\infty  {{x^{v - 1}}{e^{ - \mu x}}\ln xdx}  = {\mu ^{ - v}}\Gamma \left( v \right)\left( {\psi \left( v \right) - \ln \mu } \right)$~\cite[(4.352.1)]{gradshteyn}. Then, $\mathbb{\mathrm{E}}\left\{ {\ln \left( {\left| {{X_{o,{\rm{M}}}}} \right|} \right)} \right\}$ is derived as
\begin{align}\label{MBS_Dis_Z1}
\mathbb{\mathrm{E}}\left\{ {\ln \left( {\left| {{X_{o,{\rm{M}}}}} \right|} \right)} \right\}&\mathop = \limits^{\left(b\right)}  \int_0^\infty  {\ln \left( x \right){f_{\left| {{X_{o,{\rm{M}}}}} \right|}}\left( x \right)dx} \nonumber\\
&= \int_0^\infty  {\ln \left( x \right) \frac{2 \pi \lambda _\mathrm{M} }{{\mathcal{A}_\mathrm{M}}}x{e^{ - \pi \left( {{\lambda _\mathrm{R}} + {\lambda _\mathrm{M}}} \right){x^2}}} dx}\nonumber\\
&= \frac{1}{2}\left( {\psi \left( 1 \right) - \ln \left( {\pi \left( {{\lambda _{\rm{R}}} + {\lambda _{\rm{M}}}} \right)} \right)} \right).
\end{align}
In (b) above,  ${f_{\left| {{X_{o,{\rm{M}}}}} \right|}}\left( x \right)$ is the PDF of the distance between the typical MBS and its intended user, which can be directly  obtained following \eqref{A1_3}, and ${\mathcal{A}_\mathrm{M}}=\frac{\lambda _\mathrm{M}}{\lambda _\mathrm{R}+\lambda _\mathrm{M}}$ is the probability that
 a user is associated with the MBS. By substituting \eqref{Exp_ln_Z1} and \eqref{MBS_Dis_Z1} into \eqref{Z_1_theo2}, we obtain $Z_3$ as
 \begin{align}\label{Z_3_exp}
 Z_3=&\ln \left( {\frac{{{P_{\rm{M}}}}}{S}\beta } \right)+ \psi\left(N_\mathrm{M}-S+1\right)\nonumber\\
 &-\frac{\eta _{\rm{M}}}{2}\left( \psi \left( 1 \right)
 - \ln \left( {\pi \left( {{\lambda _{\rm{R}}} + {\lambda _{\rm{M}}}} \right)} \right) \right).
 \end{align}

From \eqref{C1_4},  considering the convexity of $\ln\left(\frac{1}{1+x}\right)$ and using Jensen's inequality, we derive the lower bound on the $Z_4$ as
\begin{align}\label{Z_2_step}
Z_4\geq\bar{Z}_4=\ln\left(\frac{1}{\mathbb{\mathrm{E}}\left\{J_{\mathrm{M},\nu}\right\}+\mathbb{\mathrm{E}}\left\{J_{\mathrm{R},\nu}\right\}+B_o N_1}\right).
\end{align}
Then, we have
\begin{align}\label{Exp_I_M}
\mathbb{\mathrm{E}}\left\{J_{\mathrm{M},\nu}\right\}= &\int_0^\infty \mathbb{\mathrm{E}}\left\{\sum\limits_{\ell  \in {\Phi_\mathrm{M}}/{o}} {\frac{P_\mathrm{M}}{S}{g_{\ell,\nu}}\beta \left| {{X_{\ell ,\mathrm{M}}}} \right|^{-\eta_\mathrm{M}}}\right\} {f_{\left| {{X_{o,{\rm{M}}}}} \right|}}\left( x \right)dx\nonumber\\
\mathop = \limits^{\left(c\right)}& \int_0^\infty \left({P_\mathrm{M}}\beta 2\pi {\lambda _\mathrm{M}} \int_x^\infty r^{1-\eta_\mathrm{M}}  dr \right) {f_{\left| {{X_{o,{\rm{M}}}}} \right|}}\left( x \right)dx\nonumber\\
=&{\frac{{{P_\mathrm{M}}\beta 2\pi {\lambda _\mathrm{M}}} \Gamma\left(2-\frac{\eta_\mathrm{M}}{2}\right) }{{\left({\eta_\mathrm{M}} - 2\right)\left(\pi\lambda _\mathrm{M}+\pi\lambda _\mathrm{R}\right)^{1-\frac{\eta_\mathrm{M}}{2}}}}},
\end{align}
where (c) results from using Campbell's theorem~\cite{Baccelli2009}. Likewise, $\mathbb{\mathrm{E}}\left\{J_{\mathrm{R},\nu}\right\}$ is calculated as
\begin{align}\label{Exp_I}
\mathbb{\mathrm{E}}\left\{J_{\mathrm{R},\nu}\right\}=& \int_0^\infty \mathbb{\mathrm{E}}\left\{\sum\limits_{j  \in {\Phi_\mathrm{R}}} {{P_\mathrm{R}}{g_{j,\nu} }\beta\left| {{X_{j ,\mathrm{R}}}} \right|^{-\eta_\mathrm{R}}}\right\} {f_{\left| {{X_{o,{\rm{M}}}}} \right|}}\left( x \right)dx\nonumber\\
=&{\frac{{{P_\mathrm{R}}\beta 2\pi {\lambda _\mathrm{R}}} \Gamma\left(2-\frac{\eta_\mathrm{R}}{2}\right) }{{\left({\eta_\mathrm{R}} - 2\right)\left(\pi\lambda _\mathrm{M}+\pi\lambda _\mathrm{R}\right)^{1-\frac{\eta_\mathrm{R}}{2}}}}}.
\end{align}
 Substituting \eqref{Z_3_exp} and \eqref{Z_2_step} into \eqref{C1_2}, we obtain \eqref{MBS_EC_nu}.

\section*{Appendix D: A detailed derivation of Corollary 1}
\label{App:theo_1}
\renewcommand{\theequation}{D.\arabic{equation}}
\setcounter{equation}{0}

When the connection outage constraint $P_{\mathrm{R},k}^{\mathrm{co}}\left(R\right)=\sigma$, using \eqref{RRH_COP_1}, we can easily get \eqref{DL_RRH_1}.

For connection outage constraint on the RRH transmission over the $\nu$-th RB shared by RRHs and MBSs, namely $P_{\mathrm{R},\nu}^{\mathrm{co}}\left(R\right)=\sigma$, we have
\begin{align}\label{D1_1}
{\bar{F}}_{{\gamma_{\mathrm{R},\nu}}\mid \left\{\left| {{X_{o,{\rm{R}}}}} \right|=d_o\right\}}\left(2^{R}-1\right)=1-\sigma.
\end{align}
Since the noise can be ignored compared with the inter-tier interference from MBSs, based on \eqref{A1_5} and \eqref{A1_6}, we consider the worse case that interferers are located everywhere in the plane and  derive the lower bound for ${\bar{F}}_{{\gamma_{\mathrm{R},\nu}}\mid \left\{\left| {{X_{o,{\rm{R}}}}} \right|=d_o\right\}}\left(\cdot\right)$ as
\begin{align}\label{C1_Eq4}
&{\bar{F}}_{{\gamma_{\mathrm{R},\nu}}\mid \left\{\left| {{X_{o,{\rm{R}}}}} \right|
=d_o\right\}}^{\mathrm{L}}\left(\gamma\right)\nonumber\\
&=\exp\left\{{ - \int_0^\infty  {\bigg( {1 - \frac{1}{{{{\left( {1 + \frac{{{P_{\mathrm{M}}}d_o^{\eta_\mathrm{R}} \gamma}}
{{P_\mathrm{R}} S} {r^{ - {\eta_\mathrm{M}}}}} \right)}^S}}}} \bigg){\lambda_\mathrm{M}}2\pi rdr} }\right\}\nonumber\\
&=\exp \bigg(  - 2\pi {\lambda _\mathrm{M}}{{\left( {\frac{{{P_\mathrm{M}}{d_o^{{\eta _\mathrm{R}}}}}}{{{P_\mathrm{R}}S}}\gamma } \right)}^{\frac{2}{{{\eta _\mathrm{M}}}}}}\sum\limits_{\mu  = 1}^S {
S\choose
\mu} \nonumber\\
&\quad \quad\quad\quad \frac{{\Gamma \left( {\mu  - \frac{2}{{{\eta _\mathrm{M}}}}} \right)\Gamma \left( { - \mu  + \frac{2}{{{\eta _\mathrm{M}}}} + S} \right)}}{{{\eta _\mathrm{M}}\Gamma \left( S \right)}}  \bigg).
\end{align}
Substituting \eqref{C1_Eq4} into \eqref{D1_1}, after some manipulations, we obtain \eqref{coro_RRH_rate2}.

\end{document}